%% file: everlasting-qkd.tex
\documentclass[runningheads,a4paper,envcountsame]{llncs}
\pdfoutput=1
\usepackage[T1]{fontenc}
\usepackage[ascii]{inputenc}
\usepackage[bookmarksnumbered,linktocpage,hypertexnames=true,colorlinks=true,linkcolor=blue,urlcolor=blue,citecolor=blue,anchorcolor=green,breaklinks=true,pdfusetitle]{hyperref}
\usepackage{xcolor}
\usepackage{physics}
\usepackage{graphicx}
\usepackage{amsmath,amssymb}
\usepackage{mathtools}
\usepackage{bbm}
\usepackage{mleftright}
\usepackage[capitalize,nameinlink]{cleveref}
\usepackage{url,comment}
\usepackage{enumitem}
\usepackage{float,xspace}
\usepackage{textcomp}

\usepackage{amsthm}
\usepackage[normalem]{ulem}
\usepackage[english]{babel}

\newcommand{\poly}{\mathsf{poly}}
\newcommand{\ot}{\otimes}
\usepackage{color}

\allowdisplaybreaks[4]

\title{Computational Monogamy of Entanglement and Non-Interactive Quantum Key Distribution}

\titlerunning{Computational Monogamy of Entanglement \& Non-Interactive QKD}
\date{}
\author{Alex B. Grilo\inst{1} \and Giulio Malavolta\inst{2} \and Michael Walter\inst{3,4,5} \and Tianwei Zhang\inst{3,6}}
\authorrunning{A.B.~Grilo, G.~Malavolta, M.~Walter, T.~Zhang}
\institute{Sorbonne Universit\'e, CNRS, LIP6 \and Bocconi University \and Faculty of Computer Science, Ruhr University Bochum \and Korteweg-de Vries Institute for Mathematics and QuSoft, University of Amsterdam \and Faculty of Physics, Ludwig Maximilian University of Munich \and Max Planck Institute for Security and Privacy}
\hypersetup{pdfauthor={A.B. Grilo, G. Malavolta, M. Walter, T. Zhang}}

\begin{document}
\maketitle
\begin{abstract}\input{abstract}\end{abstract}
\input{intro}
\input{tech}
\input{prelim}
\input{moe}
\input{protocol}
\input{nogo}

\subsection*{Acknowledgments}

GM and MW are supported by the European Union (ERC Starting Grant ObfusQation, 101077455 and SYMOPTIC, 101040907).
GM, MW and TZ acknowledge Deutsche For\-schungs\-gemeinschaft (DFG, German Research Foundation) under Germany's Excellence Strategy - EXC 2092 CASA - 390781972.
MW is also supported by the German Federal Ministry of Research, Technology and Space (QuSol, 13N17173).
ABG is supported by ANR JCJC TCS-NISQ ANR-22-CE47-0004.

\bibliographystyle{splncs04}
\bibliography{everlasting-qkd}

\end{document}

%% file: abstract.tex
% !TEX root = everlasting-qkd.tex
Quantum key distribution (QKD) enables Alice and Bob to exchange a secret key over a public, untrusted quantum channel.
Compared to classical key exchange, QKD achieves \emph{everlasting security}: after the protocol execution the key is secure against adversaries that can do unbounded computations.
On the flip side, while classical key exchange can be achieved non-interactively (with two simultaneous messages between Alice and Bob), no non-interactive protocol is known that provides everlasting security, even using quantum information.

In this work, we make progress on this problem. Our main technical contribution is a \emph{computational} variant of the celebrated \emph{monogamy of entanglement} game, where the secret is only computationally hidden from the players, rather than information-theoretically. In these settings, we prove a negligible bound on the maximal winning probability over all strategies.
As a direct application, we obtain a non-interactive (simultaneous message) QKD protocol from any post-quantum classical non-interactive key exchange, which satisfies everlastingly secure \emph{assuming Alice and Bob agree on the same key}.
The protocol only uses EPR pairs and standard and Hadamard basis measurements, making it suitable for near-term quantum hardware.
We also propose how to convert this protocol into a two-round protocol that satisfies the standard notion of everlasting security.

Finally, we prove a \emph{no-go theorem} which establishes that (in contrast to the case of ordinary multi-round QKD) entanglement is necessary for non-inter\-active QKD, i.e., the messages sent by Alice and Bob cannot both be unentangled with their respective quantum memories if the protocol is to be everlastingly secure.

%% file: intro.tex
% !TEX root = everlasting-qkd.tex
%=============================================================================
\section{Introduction}
%=============================================================================
Quantum key distribution (QKD) \cite{bennett14} enables two parties, commonly referred to as Alice and Bob, to securely exchange a secret key over a public, untrusted quantum channel. In contrast to classical key exchange protocols, QKD offers two main advantages:
(i)~It requires only authenticated classical channels, which can be practically implemented using Minicrypt~\cite{impagliazzo95} computational assumptions (in contrast, it is widely believed that such assumptions are not sufficient for classical key exchange~\cite{IR89}).
(ii)~It guarantees \emph{everlasting security}: Even if an adversary becomes unbounded after the protocol execution, no information about the key is leaked.
This prevents attacks where the adversary records data to leverage future technological/algorithmic breakthroughs.

Given the fundamental nature of the problem, it is not surprising that QKD has become one of the most well-studied topics in the theory of quantum information \cite{MY98,LC99,mayers01,renner08,MW24} and in the experimental community \cite{JKL+13,KLH+15,YCL+17}.
It is known that three messages are sufficient for building QKD~\cite{ShorPreskill}.
A recent work~\cite{MW24} achieved the first two-message protocol for QKD with everlasting security, assuming the existence of (quantum-secure) one-way functions.\footnote{The same work also shows that computational assumptions are necessary in the two-message settings.}
Two messages are optimal for QKD, but in their protocol, Bob has to send his message after receiving the message from Alice. Therefore, their protocol still requires two rounds of communication -- in contrast to classical key exchange, which can be achieved \emph{non-interactively}, that is, using a single round of two \emph{simultaneous messages} between Alice and Bob~\cite{DH76}
This prompts the question:
\begin{center}
    \emph{Can quantum protocols match the round complexity of classical protocols, \\ while still achieving everlasting security?}
\end{center}
 The purpose of this work is to make progress on this question.
\subsection{Our Results}

In this work we consider the problem of \emph{non-interactive QKD}: we seek a protocol between Alice and Bob that consists of a \emph{single round} of simultaneous messages where, at the end of the interaction, Alice and Bob agree on a secret key. We consider an attacker that is computationally bounded during the execution of the protocol, but afterwards can perform arbitrary (computationally unbounded) computations.
In this setting, we present both positive and negative results.

\paragraph{Constructions.}
On the positive side, we show how to construct a non-interactive QKD protocol from any post-quantum classical non-interactive key exchange (NIKE).
The latter can be achieved from a variety of assumptions, including the hardness of the learning with errors (LWE) problem \cite{GdKQ+24} or of computational problems related to isogenies in elliptic curves \cite{CLM+18}.
Our protocol satisfies a weak notion of \emph{everlasting security}: roughly speaking, it everlasting security holds provided Alice and Bob agree on the same shared key.\footnote{More precisely we show a notion of \emph{search} hardness, i.e., we prove that the shared key of Alice and Bob is hard to guess, conditioned on Alice and Bob agreeing on the same key. We keep this aspect deliberately informal at this point, and we will make things more precise in the subsequent sections.}
Furthermore, in our protocol, only the message sent from Alice to Bob is quantum, while Bob's message is entirely classical.

Our security proof relies on a \emph{computational} variant of the \emph{monogamy-of-entanglement game} of~\cite{MOE}.
While in the original game a random basis choice~$\theta$ is informationally hidden until the parties have agreed on a quantum state, in our game the basis choice is only \emph{computationally} hidden (that is, it is only hidden for efficient algorithms).
The game proceeds by Alice, Bob, and Charlie jointly applying an efficient algorithm that prepares a shared quantum state of their systems~$ABC$.
Then Alice and Bob measure $A$ and $B$ in the $\theta$-basis to obtain outcomes~$K_A$ and $K_B$, respectively, while Charlie is allowed to apply an arbitrary (possibly inefficient) measurement to obtain outcome~$K_C$.
The players win if $K_A = K_B = K_C$.
We describe the game more formally in the technical outline below (\cref{subsec:tech}).
Our main technical contribution is the following theorem, which can be understood as a computational monogamy of entanglement result:

\begin{theorem}[Informal]\label{thm:intro moe}
If $\theta$ is computationally hidden, the winning probability of the players in the above-described computational monogamy-of-entanglement game is negligible.
\end{theorem}

We believe that this result may be of independent interest and find other applications.
As a corollary, we obtain the following non-interactive QKD protocol:

\begin{theorem}[Informal]\label{thm:intro niqkd}
    Assuming the hardness of the LWE problem (or any other assumption that implies the existence of a post-quantum NIKE), there exists a non-interactive QKD protocol that offers everlasting security when Alice and Bob agree on the same key.
\end{theorem}

Thus, we identify a natural and meaningful setting under which truly non-interactive QKD is possible to achieve with everlasting security, which was not known prior to our work.
At the quantum level, our protocol only uses EPR pairs, and Alice and Bob measure their state as soon as they receive each other's message.
This makes our protocol a plausible candidate for experimental validation, using existing or near-term quantum hardware.

We furthermore propose how to achieve the standard notion of everlasting security by a two-round simultaneous-message protocol that builds on top of our non-interactive protocol and has essentially the same complexity.
In particular, our two-round protocol still only uses EPR pairs -- in contrast to~\cite{MW24} which used entangled states of $\poly(\lambda)$ many qubits.
However, while the protocol of~\cite{MW24} uses only two messages in total and only assumes the existence of post-quantum one-way functions, our new two-round protocol uses four messages in total and requires the existence of a post-quantum NIKE, which is considered a stronger assumption.
The question of the existence of a non-interactive QKD scheme (with two simultaneous messages, one from Alice and one from~Bob) that achieves the standard notion of everlasting security, as posed in \cite{MW24}, remains~open.

\paragraph{No-go Result.}
While traditional QKD protocols can achieve security by sending single-qubit states, our non-interactive protocol requires Alice to create EPR pairs and store one qubit of each pair until she receives Bob's message.
While experimentally more challenging, we show that this to some extent unavoidable:

\begin{theorem}[Informal]\label{thm:intro nogo}
A perfectly-correct non-interactive QKD protocol can only be everlastingly secure if it uses entanglement.
\end{theorem}

We prove this result by exhibiting an attack that does not disturb the quantum states and allows the attacker to learn the key with constant probability.

%% file: tech.tex
% !TEX root = everlasting-qkd.tex
%-----------------------------------------------------------------------------
\subsection{Technical Outline}\label{subsec:tech}
%-----------------------------------------------------------------------------

\paragraph{Non-Interactive QKD and Weak Everlasting Security.}
Before explaining our approach, let us make the scenario more concrete.
We consider a setting where Alice and Bob exchange a single round of simultaneous messages, each consisting of a classical and quantum part.
While all classical messages are delivered honestly, to model the presence of an authenticated classical channel, the quantum channel is fully untrusted:
the attacker can apply an arbitrary quantum polynomial time (QPT) channel to manipulate the quantum messages (and entangle them with his own register) before they get delivered to Alice and Bob.
Once the protocol is completed, i.e., Alice and Bob have derived their local key~$K_A$ and~$K_B$, the attacker becomes computationally unbounded and can perform arbitrary computations in order to try to guess the key.
We say that the attacker succeeds if their guess is correct and furthermore $K_A = K_B$ (Alice and Bob agree on the same key).
We define \emph{weak everlasting security} to mean that attackers succeed only with negligible probability.

\paragraph{The Protocol Blueprint.} The template for our protocol is quite natural:
We combine the celebrated QKD protocol of~\cite{bennett14} with a \emph{classical} post-quantum non-interactive key exchange (NIKE), where the key is used to select the secret basis:
\begin{itemize}
    \item \emph{Alice:} Samples a key pair using the classical NIKE protocol, and prepares $n$ EPR pairs
    $\ket{\mathsf{EPR}}^{\otimes n} = \left(\frac{\ket{00} + \ket{11}}{\sqrt{2}}\right)^{\otimes n}$,
    then she sends her classical public key, along with the second qubit of each EPR pair, to Bob.
    \item \emph{Bob:} Samples a key pair using the classical NIKE protocol and sends his public key to Alice.

    \item \emph{Outputs:} Alice uses her private key and Bob's public key to derive a classical shared key~$\theta_A \in \{0,1\}^{n}$.
    Then she measures her qubits in the $\theta_A$-basis: she measures her $j$-th qubit in the standard basis if~$\theta_{A,j}=0$ and otherwise in the Hadamard basis. % if~$\theta_{A,j}=1$.
    Alice sets~$K_A \in \{0,1\}^n$ to be the bitstring containing the measurement outcomes.

    Bob proceeds identically, by using his private key and Alice's public key to derive a classical shared key~$\theta_B$ and obtaining~$K_B$ as the measurement outcomes of the~$\theta_B$-basis measurement on his qubits.
\end{itemize}
Correctness follows from correctness of the NIKE, since if Alice and Bob agree on the same basis $\theta_A = \theta_B$, then they measure their EPR pairs in the same basis, resulting in the same outcomes.

However, proving security is much less obvious.
One standard approach would be to appeal to a \emph{monogamy of entanglement} game in the spirit of~\cite{MOE}.
We will elaborate more on this later, but for the moment it suffices to say that known statements are \emph{information-theoretic}, i.e., they crucially use the assumption that even computationally-unbounded adversaries have no information about the basis~$\theta$.
In our setting this is rather not true: Because of classical NIKE protocol is only computationally secure, the basis choice is only computationally hidden (that is, hidden from efficient quantum adversaries), but not information-theoretically so.
To use the computational security of the NIKE, we therefore need a computational argument, i.e., an \emph{efficient} reduction.
However, simple reduction strategies do not seem to work either: we cannot just switch the basis~$\theta$ to a uniform string and appeal to the security of the NIKE protocol, because in the second stage the adversary's power is \emph{unbounded}.
Therefore, running the entire adversary as part of a reduction would take the reduction unbounded time, making the security guarantees of the classical NIKE protocol not applicable.
Therefore, while our strategy ought to appeal to the computational security of the NIKE, it has to do so in an indirect manner.

\paragraph{Computational Monogamy of Entanglement.}
We formalize our solution in a more abstract scenario, by defining and analyzing a computational variant of the monogamy-of-entanglement game of~\cite{MOE}.

We assume the existence of an efficiently-sampleable distribution $\mathcal{Z}(1^\lambda)$ supported on pairs $(p, \theta) \in \mathcal{P}_\lambda \times \{0,1\}^{n(\lambda)}$ for some polynomial $n=n(\lambda)$. We require that the following distributions are computationally indistinguishable:
\begin{equation}\label{eq:intro assm}
\left((p, \theta) : (p, \theta) \gets \mathcal{Z}(1^\lambda)\right) \approx_c
\left((p, \theta^*) : (p, \cdot) \gets \mathcal{Z}(1^\lambda); \theta^* \gets \{0,1\}^n\right).
\end{equation}
The game proceeds as follows:
\begin{enumerate}
\item \emph{Sampling Phase:} Alice samples $(p, \theta) \gets \mathcal{Z}(1^\lambda)$ and reveals $p$ to Bob and Charlie.
\item \emph{Efficient Preparation Phase:} Alice, Bob, and Charlie jointly apply a QPT algorithm (with input~$p$) to create a shared quantum state between their registers $A$, $B$,~$C$.
Registers~$A$ and~$B$ should consist of $n(\lambda)$ qubits, while $C$ can be arbitrary.
\item \emph{Question Phase:} Alice measures register~$A$ in the $\theta$-basis to obtain an outcome~$K_A$.
She then reveals $\theta$ to Bob and Charlie.
\item \emph{Semi-Honest Answer Phase:} Bob measures register~$B$ in the $\theta$-basis to obtain an outcome~$K_B$, while Charlie can apply an arbitrary (possibly inefficient) measurement of register~$C$ to obtain an outcome~$K_C$.
\end{enumerate}
The players win the game if~$K_A = K_B = K_C$.

\medskip

The key differences to the original monogamy-of-entanglement game of \cite{MOE} are as follows:
Most significantly for us, in our game Bob and Charlie have some information~$P$ about~$\theta$ \emph{before} creating their shared entangled state, which is not the case in \cite{MOE}.
On the other hand, we require the shared state to be efficiently preparable, and we also assume that Bob's measurement in the answer phase is performed honestly, whereas in \cite{MOE} is an arbitrary POVM.

Note that this game is a good model for the non-interactive QKD protocol described above:
Charlie essentially plays the role of the attacker;
in the preparation phase, Alice creates $n$ EPR pairs, and Charlie applies an arbitrary efficient quantum channel to create systems $B$ and $C$.

To prove a bound on the success probability in the above game, we consider the following thought experiment: Let $\rho_{AB}$ be the joint state of Alice and Bob, right before the question phase.
We split the state into $n/s$ blocks, each of size $s$ (we have some freedom in the choice of parameters, but for this overview it suffices to take $s = \sqrt{n}$).
Then we imagine applying the binary POVM $\{M_{AB}^{(0)}, M_{AB}^{(1)}\}$ given by
\begin{align*}
    M_{AB}^{(1)} = \left(I -\ketbra{\mathsf{EPR}}{\mathsf{EPR}}^{\otimes s}\right)^{\otimes \frac{n}{s}},
    && M_{AB}^{(0)} = I - M_{AB}^{(1)}.
\end{align*}
If the measurement outcome is~$0$, then, \emph{roughly speaking}, we project the state $\rho_{AB}$ onto a state that has \emph{at least}~$s$ EPR pairs shared between Alice and Bob.
Indeed we can prove that in this case, no matter what Charlie does, his probability of guessing Alice and Bob's outcomes (assuming they agree) is bounded by $\tilde O(2^{-s}) = \tilde O(2^{-\sqrt{n}})$.

To complete the proof, we need to bound the probability that the game is won if the above-described POVM returns outcome~$0$.
In fact, we can show something stronger:
In this case the probability that Alice and Bob agree is negligible (this is stronger because $K_A = K_B$ is a necessary but not sufficient condition for winning the game).
To see this, let us assume for a moment that Alice and Bob perform measurements in a basis~$\theta^*$ sampled uniformly at random and independently from~$p$.
In this case, the probability that Alice and Bob agree is given by
\[
    \Tr\left(\mathsf{E}_{\theta^*\in \{0,1\}^n}\left(\sum_{x\in \{0,1\}^n} {}_{\theta^*}\!\ketbra{xx}{xx}_{\theta^*}\right)\left(I-(\ketbra{\mathsf{EPR}}{\mathsf{EPR}})^{\otimes s}\right)^{\otimes n/s}\rho_{AB}^{(p)}\right),
\]
where $\ket x_\theta$ denotes the basis states in the~$\theta$-basis.
A direct calculation shows that this probability can be bounded by $2^{-n/s} = 2^{-\sqrt n}$, independently of the quantum state.
In the actual experiment Alice and Bob measure according to $\theta$, which is correlated with~$p$.
However, the probability of agreement cannot differ from the case treated above, as otherwise we could efficiently distinguish~$(p,\theta)$ from~$(p,\theta^*)$, in contradiction to \cref{eq:intro assm}.
This is the point where we finally appeal to the computational indistinguishability of the two distributions.
Crucially, this reduction only uses the efficiently prepared state of Alice and Bob, while Charlie's later unbounded computation does not enter the picture.

This concludes the analysis of the monogamy-of-entanglement game (\cref{thm:intro moe}).
It is not hard to obtain from this desired security of the non-interactive QKD protocol (\cref{thm:intro niqkd}).

\paragraph{Achieving Everlasting Security in Two Rounds.}
In order to obtain a QKD protocol with the standard notion of everlasting security (indistinguishability also in case of disagreement), we propose adding another round of simultaneous messages, where Alice and Bob test the equality of their key.
To achieve this without leaking too much information we consider a standard technique:
instead of sending~$K_A$ and~$K_B$ in the plain, Alice and Bob send and compare hashes of their respective keys. With overwhelming probability, this test fails if~$K_A \neq K_B$.
Finally, to turn search security (the key is hard to guess) to indistinguishability from random we take a quantum-proof randomness extractor, seeded by the XOR of two seeds sampled independently by Alice and Bob.

\paragraph{Entanglement is Necessary.}
We prove our impossibility result (\cref{thm:intro nogo}) by showing that if there is no entanglement, the key shared by Alice and Bob in an honest run of the protocol is a function of the classical randomness held by Alice and Bob. In particular, this implies that the honest measurements of the protocol are non-destructive: they do not collapse the quantum~messages!
An attack can then proceed as follows. Eve intercepts the message by Alice, simulates polynomially many possible runs of Bob, and computes simulated key for each run.
Similarly, Eve intercepts the message by Bob, simulates polynomially many possible runs of Alice, and computes the simulated key that would be output in each run.
After collecting this data, Eve forwards Alice's message to Bob and vice-versa, who continue the protocol.
Because the measurements are non-destructive, from the perspective of Alice and Bob they are in an honest run of the protocol.
On the other hand, the data collected by Eve allows her to later guess classical randomness such that the resulting key matches the key of Alice and Bob with constant probability.

\subsection{Open Problems}

As already mentioned above, the question of a non-interactive QKD satisfying the standard definition of everlasting security remains open from \emph{any} computational assumption.
We suspect that new ideas are needed to construct such a protocol (or to rule out its existence).
Another interesting open problem is to exhibit two-round protocols with a positive key rate.
At a more fundamental level, a fascinating direction is to strengthen our computational monogamy of entanglement result: It is conceivable that a bound can be proven even without restricting Bob to measure its state in the $\theta$-basis, instead allowing him to apply an arbitrary POVM.
Both the setting of a polynomial-time POVM and the setting of a computationally unbounded one are open, although we expect the latter to be difficult to prove given know techniques, since a polynomial-time reduction would not even be able to run Bob's algorithm.

%% file: prelim.tex
% !TEX root = everlasting-qkd.tex
%=============================================================================
\section{Preliminaries}
%=============================================================================

Throughout this work, we denote the security parameter by $\lambda$.
We denote by~$1^{\lambda}$ the all-ones string of length~$\lambda$.
We say that a function~$f$ is negligible in the security parameter~$\lambda$ if~$f(\lambda) = {\lambda}^{-\omega(1)}$ or, equivalently, $f(\lambda) = 2^{-\omega(\log\lambda)}$; often such a function is simply denoted~$\mathsf{negl}$.

For a finite set $S$, we write~$x \leftarrow S$ to denote that~$x$ is sampled uniformly at random from~$S$, and for a probability distribution~$\mathbb{P}$, we write~$x \leftarrow \mathbb{P}$ to denote that~$x$ is sampled according to~$\mathbb{P}$.
Unless stated otherwise, all random variables and probability distributions are finitely supported, that is, take values in finite sets.
We denote by~$[n]$ the set~$\{1, \dots, n\}$.
We write~$I$ for identity matrices or operators, and~$\Tr$ for the trace of a matrix or operator.
A unitary operator~$U$ is one that satisfies~$UU^\dagger = U^\dagger U=I$, and a Hermitian operator~$H$ is one such that~$H^\dagger = H$.

%-----------------------------------------------------------------------------
\subsection{Quantum Information}
%-----------------------------------------------------------------------------
In this section, we provide a brief overview of quantum information. For a more detailed introduction, see~\cite{NC16,watrous2018theory}.
A \emph{(quantum) register}~$A$ consisting of $n$ qubits is associated with the Hilbert space $\mathcal{H}_A = (\mathbb{C}^2)^{\otimes n}$. Given two registers $A$ and $B$, we denote the composite register by~$AB$.
The corresponding Hilbert space is given by the tensor product $\mathcal{H}_{AB} = \mathcal{H}_A \otimes \mathcal{H}_B$.

\paragraph{Quantum States.}
The \emph{(quantum) state} of a register $A$ is described by a density operator $\rho_A$ on $\mathcal{H}_A$, which is a positive semi-definite Hermitian operator with trace equal to one.
A state is called \emph{pure} if it has rank one.
Thus, pure quantum states can be represented by unit vectors $\ket\psi_A \in \mathcal{H}_A$, with $\rho_A = \ketbra{\psi}{\psi}_A$.
For a quantum state~$\rho_{AB}$ on $\mathcal{H}_{AB}$, we denote $\rho_A = \Tr_B(\rho_{AB})\in \mathcal{H}_A$ the reduced state of $\rho_{AB}$ on $A$.

The quantum formalism allows treating classical and quantum information on the same footing.
For example, if $X$ is a random variable with outcomes in some set~$\mathcal X$, its probability distribution can be described by the \emph{classical} quantum state~$\rho_X = \sum_{x \in \mathcal X} p_x \ketbra x$, where $p_x = \Pr(X = x)$.
For the uniform distribution, this called the \emph{maximally mixed state}~$\tau_X = \frac1{\lvert\mathcal X\rvert} \sum_{x \in \mathcal X}\ketbra{x}{x}$.
More generally, if we have a random variable~$X$ and a quantum register~$E$ such that $E$ is in state~$\rho_E^{(x)}$ conditional on $X=x$, this can be described by a \emph{classical-quantum (cq) state}~$\rho_{XE} = \sum_x p_x \ketbra x \otimes \rho_E^{(x)}$.
In the above, subscripts indicate the registers (we only omit them when the context is clear).

For a quantum state $\rho_{AB}$ on two registers~$A$ and~$B$, we often denote by~$\rho_A = \tr_B[\rho_{AB}]$ for the reduced state of register~$A$.
Dually, if $M_A$ is an operator (typically a unitary, a projection, or a POVM element, see below) on register~$A$, we extend it implicitly by the identity to an operator~$M_A \ot I_B$.
These notations are compatible:
we have $\tr[M_A \rho_{AB}] = \tr[M_A \tr_B[\rho_{AB}]] = \tr[M_A \rho_A]$.

\paragraph{Quantum Channels and Measurements.}
A \emph{(quantum) channel} $\mathcal{F}$ is a completely positive trace-preserving (CPTP) map from a register $A$ to a register $B$.
In other words, given any density matrix $\rho_A$, the channel $\mathcal{F}$ produces $\mathcal{F}(\rho_A) = \sigma_B$, which is another state on register~$B$, and the same applies when $\mathcal{F}$ is applied to the $A$-register of a quantum state $\rho_{AC}$, resulting in the quantum state~$\sigma_{BC} = (\mathcal{F} \ot \mathcal{I})(\rho_{AC})$, where $\mathcal{I}$ denotes the identity channel.
For any unitary operator~$U$, there is a quantum channel~$\mathcal U$ that maps any input state~$\rho$ to the output state $\mathcal{U}(\rho) := U \rho U^\dagger$.

A \emph{projective measurement} is defined by a set of projectors $\{ \Pi_j \}_j$ such that $\sum_j \Pi_j = I$.
A projector $\Pi$ is a Hermitian operator such that $\Pi^2 = \Pi$, that is, an orthogonal projection.
Given a state $\rho$, the measurement yields outcome $j$ with probability $p_j = \Tr(\Pi_j \rho)$, upon which the state changes to ${\Pi_j \rho \Pi_j/p_j}$.
A basis measurement is one where~$\Pi_j=\ketbra{e_j}$ and the $\{\ket{e_j}\}$ (necessarily) form an orthonormal basis.

A \emph{positive operator-valued measure (POVM)} is a generalization of a projective measurement.
A POVM is defined by a set of positive semi-definite operators~$\{E_j\}_j$ such that~$\sum_j E_j = I$ (that is, the $E_j$ no longer need to be projections).
As before, given a quantum state~$\rho$, the probability of obtaining outcome $j$ when performing the measurement is given by~$p(j) = \Tr(E_j \rho)$, but the state after the measurement is no longer uniquely specified.
Indeed, while any POVM measurement can be realized by a projective measurement on a larger Hilbert space (by Naimark's dilation theorem), different realizations can lead to different post-measurement states.
A \emph{binary POVM} is one that has two outcomes~0 and~1.
Binary POVMs, are in one-to-one correspondence with quantum channels that output a single bit (i.e., the output state is a mixture of $\ketbra0$ and $\ketbra1$ for any input state).

\begin{lemma}[Operator Union Bound]\label{lmm:IminusPitensor}
Let $P_1,\dots,P_t$ be PSD operators such that~$I - P_i$ is also PSD for all~$i\in[t]$.
Then:
\begin{align*}
    I - \bigotimes_{i=1}^t P_i
\leq \sum_{i=1}^t \left( I^{\otimes (i-1)} \otimes (I - P_i) \otimes I^{\otimes (t-i)} \right)
\end{align*}
\end{lemma}
\begin{proof}
For $t = 1$, the statement is trivially true.
We now prove this by induction on $k$, so let us assume that the statement is true for some value $k$.
We will prove it also holds for $k+1$:
    \begin{align*}
        I - \bigotimes_{i=1}^{k+1} P_i
    &=   I - \left( \bigotimes_{i=1}^k P_i \right) \otimes P_{k+1} \\
    &=   \left( I - \bigotimes_{i=1}^k P_i \right) \otimes I
    + \left( \bigotimes_{i=1}^k P_i \right) \otimes(I - P_{k+1}) \\
    &\leq   \left( I - \bigotimes_{i=1}^k P_i \right) \otimes I
    + I^{\otimes k} \otimes(I - P_{k+1}) \\
    &\leq \sum_{i=1}^k \left( I^{\otimes(i-1)} \otimes(I - P_i) \otimes I^{\otimes(k-i)} \right) \otimes I + I^{\otimes k} \otimes(I - P_{k+1}) \\
    &= \sum_{i=1}^k \left( I^{\otimes(i-1)} \otimes(I - P_i) \otimes I^{\otimes(k+1-i)} \right).
    \end{align*}
    The last inequality is by the induction hypothesis.
\end{proof}

\paragraph{Computational Basis, Hadamard Basis, and Bell Basis.}
For a single qubit, the \emph{computational basis} is denoted by
$\ket0 = \begin{psmallmatrix}1\\0\end{psmallmatrix}$
and
$\ket1 = \begin{psmallmatrix}0\\1\end{psmallmatrix}$,
while the \emph{Hadamard basis} is given by
$\ket+ = H \ket0 = \frac{\ket0 + \ket1}{\sqrt2}$
and
$\ket- = H \ket1 = \frac{\ket0 - \ket1}{\sqrt2}$.
Here,
$H = \frac1{\sqrt2}\begin{psmallmatrix}1 & 1 \\ 1 & -1\end{psmallmatrix}$
is the Hadamard unitary.
Note that $\ket0,\ket1$ is an eigenbasis of the Pauli $Z$-operator, while $\ket+,\ket-$ is an eigenbasis of the Pauli $X$-operator.
These operators are the unitaries defined by
$X = \begin{psmallmatrix}
0 & 1 \\
1 & 0
\end{psmallmatrix}$
% $Y =
% \begin{psmallmatrix}
% 0 & -i \\
% i & 0
% \end{psmallmatrix}$,
and
$Z = \begin{psmallmatrix}
1 & 0 \\
0 & -1
\end{psmallmatrix}$,
and they are also Hermitian, so that $X^2 = Z^2 = I$.

For more than one qubit, we can choose either basis for each qubit:

\begin{definition}[$\theta$-Basis States]
We denote, for $x, \theta \in \{0,1\}^n$,
\begin{align*}
    \ket x_\theta = H^\theta \ket x,
\quad\text{where}\quad
    \ket x = \ket{x_1} \otimes \cdots \otimes \ket{x_n}
\;\text{and}\;
    H^\theta = H^{\theta_1} \otimes \cdots \otimes H^{\theta_n},
\end{align*}
where we use the notation $H^1 = H$ and $H^0 = I$, with $I$ the identity matrix.
The basis~$\{\ket x_\theta\}_{x\in\{0,1\}^n}$ is called the \emph{$\theta$-basis}.
\end{definition}

Thus $\theta$ labels the basis choice and $x$ the state with respect to the chosen basis.
For example, $\ket{01}_{10} = H\ket0 \otimes \ket1 = \ket+ \otimes \ket1$.

\bigskip

For two qubits, we not only have the product bases discussed earlier but also an important basis known as the \emph{Bell basis}.
It consists of the four maximally entangled \emph{Bell states}:
\begin{equation*}
    \begin{aligned}
    &\quad\ \ket{\phi^{+}} =\ket{\mathsf{EPR}} = \frac{1}{\sqrt{2}}(\ket{00}+\ket{11}),
    &\quad\ \ket{\psi^{+}} = \frac{1}{\sqrt{2}}(\ket{01}+\ket{10}),\\
    &\quad\ \ket{\phi^{-}} = \frac{1}{\sqrt{2}}(\ket{00}-\ket{11}),
    &\quad\ \ket{\psi^{-}} = \frac{1}{\sqrt{2}}(\ket{01}-\ket{10}).
    \end{aligned}
\end{equation*}
The Bell states form a joint eigenbasis of the two-qubit Pauli operators~$X \otimes X$ and~$Z \otimes Z$, and they are uniquely characterized by the corresponding eigenvalues.
In particular, $(X \otimes X) \ket{\phi^+} = (Z \otimes Z) \ket{\phi^+} = \ket{\mathsf{EPR}}$.
It follows that if one measures both qubits of an EPR pair in the standard basis, or both in the Hadamard basis, then the outcomes always coincide.
Furthermore:

\begin{lemma}[Support of EPR Pairs]\label{lem:support}
Let $\theta \in \{0,1\}$ and $P_\theta = \sum_{x=0}^1 {}_\theta\!\ketbra{xx}_\theta$, with $\ket{xx}_\theta := \ket x_\theta \ket x_\theta$.
Then, \[ P_\theta \ket{\phi^+} = \ket{\phi^+}. \]
\end{lemma}
\begin{proof}
Because the EPR pair is invariant under Hadamard gates on both qubits, $(H \ot H) \ket{\phi^+} = \ket{\phi^+}$, we have that $\ket{\phi^+} = \frac1{\sqrt 2} \sum_{x=0}^1 \ket{xx}_\theta$ for any $\theta\in\{0,1\}$.
\end{proof}

\paragraph{Statistical and Computational Distinguishability.}
The \emph{trace distance} between two states $\rho$ and $\sigma$ is defined as:
\[
\mathsf{Td}(\rho, \sigma) = \frac{1}{2}\norm{\rho-\sigma}_1 = \frac{1}{2}\Tr\left(\sqrt{(\rho-\sigma)^\dagger(\rho-\sigma)}\right).
\]
The operational meaning of the trace distance is that
$\frac{1}{2} (1+\mathsf{Td}(\rho, \sigma))$ is the maximal probability that two states $\rho$ and $\sigma$ can be distinguished by any (not necessarily efficient) quantum channel or POVM.
That is,
\begin{align*}
    \mathsf{Td}(\rho, \sigma) = \max_{\mathcal A} \abs*{ \Pr\left( \mathcal A(\rho) = 1 \right) - \Pr\left( \mathcal A(\sigma) = 1 \right) },
\end{align*}
where the maximum is over arbitrary quantum channels~$\mathcal A$ that output a single bit.
Thus, the trace distance generalizes the statistical (total variation) distance from probability theory.
We will also use the trace distance for \emph{subnormalized states}, that is, positive semi-definite operators with trace at most one (these generalize sub-probability distributions in probability theory).

We will also consider computational indistinguishability.
To this end, recall that a \emph{nonuniform QPT algorithm} $\mathcal A = \{\mathcal A_\lambda\}$ consists of a family of quantum channels that can be implemented by polynomial-size quantum circuits that get quantum states of a polynomial number of qubits as advice.
We call $\mathcal A$ a \emph{nonuniform QPT distinguisher} if the channels output a single bit.

\begin{definition}[Computational Indistinguishability]\label{def:computational indist}
We say that two families of states~$\{\rho_\lambda\}, \{\sigma_\lambda\}$ are \emph{computationally indistinguishable}, denoted $\{\rho_\lambda\} \approx_c \{\sigma_\lambda\}$, if for every nonuniform QPT distinguisher~$\mathcal A = \{\mathcal A_\lambda\}$ there exists a negligible function~$\mathsf{negl}$ such that the following holds for all~$\lambda$:
\begin{align}\label{eq:advantage}
     \abs*{ \Pr\left( \mathcal A_\lambda(\rho_\lambda) = 1 \right) - \Pr\left( \mathcal A_\lambda(\sigma_\lambda) = 1 \right) } \leq \mathsf{negl}(\lambda).
\end{align}
The two families are called \emph{strongly computationally indistinguishable}, denoted $\{\rho_\lambda\} \approx_{sc} \{\sigma_\lambda\}$, if there exists a single negligible function $\mathsf{negl}$ such that for every nonuniform QPT distinguisher~$\mathcal A = \{\mathcal A_\lambda\}$ there exists $\lambda_0$ such that \cref{eq:advantage} holds for all~$\lambda \geq \lambda_0$.
\end{definition}

The latter, stronger notion is also a natural one~\cite{haastad1999pseudorandom}.
It applies, e.g., when more concrete bounds on the advantage of adversaries are considered.
See also the discussion below \cref{def:nike}.

Finally, we note that a \emph{(uniform) QPT algorithm} is defined as above but the quantum circuit family is uniformly generated and there is no advice.
There are also interactive definitions of both uniform and nonuniform QPT algorithms.

\paragraph{Min-Entropy and Quantum-Proof Extractors.}
The conditional min-entropy of quantum states is defined as follows~\cite{KRS09}.

\begin{definition}[Conditional Min-Entropy]
Let $\rho_{AB}$ be a quantum state. The \emph{min-entropy of $A$ conditioned on $B$} is defined by
\[
H_{\min}(A|B)_\rho := -\inf_{\sigma_B} D_\infty\!\Bigl(\rho_{AB}\,\Big\|\,I_A \otimes \sigma_B\Bigr),
\]
where the infimum is taken over all density operators $\sigma_B$ on subsystem $B$, and where
\[
D_\infty(\alpha \,\|\, \beta) := \inf\!\left\{ \lambda \in \mathbb{R} : \alpha \leq 2^\lambda \beta \right\}.
\]
\end{definition}

In the case that the first system is classical, the following theorem states that the conditional min-entropy can be interpreted as a guessing probability~\cite{KRS09}.

\begin{theorem}[Min-Entropy of classical-quantum states]\label{thm:cq-min-entropy}
Consider a classical-quantum state $\rho_{XB} = \sum_x p_x \ketbra{x}{x} \otimes \rho_B^{(x)}$.
Then,
\[
H_{\min}(X|B)_\rho = -\log p_{\text{guess}}(X|B)_\rho,
\]
where
\[
    p_{\text{guess}}(X|B)_\rho
:= \max_{\{E^{(x)}_B\}} \sum_x p_x\,\Tr(E^{(x)}_B\, \rho_B^{(x)})
= \max_{\{E^{(x)}_B\}} \Tr( \rho_{XB} \sum_x \ketbra{x}_X \ot E^{(x)}_B ).
\]
is the maximal probability of obtaining~$X$ using an arbitrary POVM $\{E^{(x)}_B\}_x$~on~$B$.
\end{theorem}

The conditional min-entropy satisfies the following chain rule~\cite[Lemma~11]{WTHR11}:

\begin{theorem}[Chain Rule]\label{chain}
Let $\rho_{ABZ}$ be a tripartite state that is classical on $Z$.
Then,
\[
H_{\min}(A|BZ)_\rho \geq
H_{\min}(A|B)_\rho - \log~\lvert Z\rvert,
\]
where $\lvert Z \rvert$ is the dimension of system $Z$ (that is, the size of the underlying classical alphabet).
\end{theorem}

Next, we recall the following definition of (quantum-proof) randomness extractor.

\begin{definition}[Extractor]\label{def:extractor}
A PPT algorithm $\mathsf{Ext} \colon \mathcal{S} \times \mathcal{X} \to \{0, 1\}^\ell$ is called a \emph{seeded strong average-case $(k,\varepsilon)$-extractor} if the following holds:
for any cq-state $\rho_{XB} = \sum_x p_x \ketbra{x}{x} \otimes \rho_B^{(x)}$ such that $H_{\min}(X|B) \geq k$, we have
\[ \mathsf{Td}(\rho_{YSB}, \tau_Y \otimes \rho_{SB}) \leq \varepsilon, \]
where
\begin{align*}
    \rho_{YSB} &= \frac1{\lvert S\rvert} \sum_{s \in \mathcal S} \sum_{x \in \mathcal X} p_x \ketbra{\mathsf{Ext}(s,x)}{\mathsf{Ext}(s,x)} \otimes \ketbra{s}{s} \otimes \rho_B^{(x)}
\end{align*}
describes the joint state of the result of the extraction (Y), the seed (S), and the quantum side information (B), and where we recall that $\tau_Y$ denotes the maximally mixed state on Y.
\end{definition}

We recall the definition of universal hash functions.

\begin{definition}[Universal Hash Family]\label{UniversalHash}
A family $\mathbbm{H} = \{ h \colon [N]\to[M] \}$ of functions is a  \emph{universal hash} if for every $x,y \in [N]$ such that $x \ne y$, it~holds~that
\[
\Pr_{h \leftarrow \mathbbm{H}} \bigl( h(x) = h(y) \bigr) = \frac{1}{M}.
\]
\end{definition}

It is well-known that \emph{efficient} constructions of universal hash families exist \cite{carter1977universal}.
Moreover, randomness extractors can be constructed from universal hash families~\cite{DORS08,renner08,konig2011sampling}.

\begin{lemma}[Generalized Leftover Hash Lemma]\label{lemma:lhl}
Let $\mathbbm{H} = \{ h \colon [N]\to\{0,1\}^\ell \}$ be a universal hash family.
Then, $\mathsf{Hash} \colon \mathbbm{H} \times [N] \to \{0,1\}^\ell$ defined by $\mathsf{Hash}(h,x) = h(x)$ is a seeded strong average-case $(k,\varepsilon)$-extractor for any $k \geq \ell + 2 \log(1/\varepsilon)$.
\end{lemma}

We also rely on the computational notion of a collision-resistant hash function, which we define next.
As in the definition of strong computational indistinguishability (\cref{def:computational indist}) we assume that there exists a single negligible function that applies to all QPT adversaries.

\begin{definition}[Collision-Resistant Hash Function]\label{def:CR}
A family $\{\mathbbm{H}_\lambda\}$ of function families is called a \emph{collision-resistant hash function} if there exists a negligible function $\mathsf{negl}$ such that the following holds:
for every QPT adversary $\mathcal{A}$ there exists $\lambda_0$ such that, for all $\lambda\geq\lambda_0$, we have
\[
\Pr_{h \leftarrow \mathbbm{H}_\lambda, (x,y)\gets\mathcal{A}(h)} \bigl( x \neq y \text{ and } h(x) = h(y)\bigr) \leq \mathsf{negl}(\lambda).
\]
\end{definition}

%-----------------------------------------------------------------------------
\subsection{Post-Quantum Non-Interactive Key Exchange}
%-----------------------------------------------------------------------------
Following~\cite{CKS08,CLM+18,FHKP13,GdKQ+24}, we formally define a post-quantum non-interactive key exchange protocol (that is, one that is computationally secure against quantum adversaries).

\begin{definition}[Post-Quantum Non-Interactive Key Exchange]\label{def:nike}
A \emph{post-quantum non-interactive key exchange (NIKE)} protocol is defined as a tuple
$\mathsf{NIKE} = (\mathsf{Stp},\mathsf{Gen},\mathsf{SdK})$ of the following algorithms, with an identity space $\mathsf{IDS}\subseteq \{0,1\}^{n(\lambda)}$ and a shared key space $\mathsf{SKS}\subseteq \{0,1\}^{n(\lambda)}$ for a polynomially bounded~$n(\lambda)$:
\begin{itemize}
    \item $\mathsf{pp} \leftarrow \mathsf{Stp}(1^{\lambda})$:
    Given the security parameter encoded in unary, $1^{\lambda}$, the PPT algorithm~$\mathsf{Stp}$ returns public system parameters $\mathsf{pp}$.
    \item $(\mathsf{sk}_A, \mathsf{pk}_A) \leftarrow \mathsf{Gen}(\mathsf{pp},A)$:
    Given the public parameters $\mathsf{pp}$ and an identity~$A \in \mathsf{IDS}$, the PPT algorithm $\mathsf{Gen}$ returns a secret-public key pair~$(\mathsf{sk}_A, \mathsf{pk}_A)$.
    \item $K\leftarrow \mathsf{SdK}(A, \mathsf{pk}_{A}, B,\mathsf{sk}_{B})$:
    Given an identity $A \in \mathsf{IDS}$ and a corresponding public key $\mathsf{pk}_{A}$, along with another identity $B \in \mathsf{IDS}$ and corresponding secret key $\mathsf{sk}_{B}$, $\mathsf{SdK}$ should be a \emph{deterministic} PPT algorithm that returns a shared key~$K \in \mathsf{SKS}$, or an abort symbol~$\perp$.
    If $A=B$ then $\mathsf{SdK}$ always returns~$\perp$.
\end{itemize}
We always assume the following two properties:
\begin{itemize}
    \item \emph{Correctness:} There exists a negligible function $\mathsf{negl}$ such that  for all~$A,B \in \mathsf{IDS}$, it holds that
    \[
    \Pr\left(\mathsf{SdK}(A, \mathsf{pk}_{A}, B,\mathsf{sk}_{B}) \neq \mathsf{SdK}(B, \mathsf{pk}_{B}, A,\mathsf{sk}_{A})\right) = \mathsf{negl}(\lambda),
    \]
    where $\mathsf{pp} \leftarrow \mathsf{Stp}(1^{\lambda})$, $(\mathsf{sk}_A, \mathsf{pk}_A) \leftarrow \mathsf{Gen}(\mathsf{pp},A)$, and $(\mathsf{sk}_B, \mathsf{pk}_B) \leftarrow \mathsf{Gen}(\mathsf{pp},B)$.

    % \item \emph{Post-Quantum Security:}
    % For all $A,B \in \mathsf{IDS}$, we have
    % \begin{equation}\label{eq:pq non-universal}
    % \left(\mathsf{pp}, \mathsf{pk}_A, \mathsf{pk}_B, \mathsf{SdK}(A, \mathsf{pk}_{A}, B,\mathsf{sk}_{B})\right) \approx_c
    % \left(\mathsf{pp}, \mathsf{pk}_A, \mathsf{pk}_B, K^*\right)
    % \end{equation}
    % where $\mathsf{pp} \leftarrow \mathsf{Stp}(1^{\lambda})$, $(\mathsf{sk}_A, \mathsf{pk}_A) \leftarrow \mathsf{Gen}(\mathsf{pp},A)$, $(\mathsf{sk}_B, \mathsf{pk}_B) \leftarrow \mathsf{Gen}(\mathsf{pp},B)$, and $K^*\leftarrow \mathsf{SKS}$,
    % % That is, the advantage of any QPT adversary in distinguishing the two distributions is negligible.
% \end{itemize}
% We also consider the following stronger security definition, where the computational indistinguishability holds for a single negligible function (and $\lambda$ large enough):
% \begin{itemize}

    \item \emph{Post-Quantum Security:}
    For all $A,B \in \mathsf{IDS}$, we have
    \begin{equation}\label{eq:pq strong}
    \left(\mathsf{pp}, \mathsf{pk}_A, \mathsf{pk}_B, \mathsf{SdK}(A, \mathsf{pk}_{A}, B,\mathsf{sk}_{B})\right) \approx_{sc}
    \left(\mathsf{pp}, \mathsf{pk}_A, \mathsf{pk}_B, K^*\right)
    \end{equation}
    where $\mathsf{pp} \leftarrow \mathsf{Stp}(1^{\lambda})$, $(\mathsf{sk}_A, \mathsf{pk}_A) \leftarrow \mathsf{Gen}(\mathsf{pp},A)$, $(\mathsf{sk}_B, \mathsf{pk}_B) \leftarrow \mathsf{Gen}(\mathsf{pp},B)$, and $K^*\leftarrow \mathsf{SKS}$.
\end{itemize}
\end{definition}

Post-quantum NIKE protocols can be constructed assuming the hardness of the standard learning with errors problem~\cite{GdKQ+24} or from computational problems in isogenies over elliptic curves~\cite{CLM+18}.
The stronger definition of computational indistinguishability used in \cref{eq:pq strong} (see \cref{def:computational indist} for the precise definition of $\approx_{sc}$) requires making concrete assumptions on the runtime of the best attacker against the underlying hard problem.
This is not unique to our settings and it is in fact required by essentially any application that considers concrete security estimates for the NIKE.
We refer the reader to~\cite{GdKQ+24,langrehr2023multi} for concrete bounds on lattice-based NIKE and to~\cite{practicalCSIDH} for isogeny-based schemes.

We remark that one can also consider a stronger definition of security \cite{CKS08}, where the adversary is given access to a key derivation oracle, for both honestly generated keys.
Since the above weaker definitions will suffice for us, we refrain from defining the stronger variant.

%% file: moe.tex
% !TEX root = everlasting-qkd.tex
%=============================================================================
\section{Computational Monogamy of Entanglement}
%=============================================================================
In this section we propose and analyze a computational variant of the monogamy of entanglement game of~\cite{MOE}.

%-----------------------------------------------------------------------------
\subsection{Definition of Computational Monogamy-of-Entanglement Game}
%-----------------------------------------------------------------------------

We assume the existence of a distribution~$\mathcal{Z}$ on $\{0,1\}^{q(\lambda)} \times \{0,1\}^{n(\lambda)}$, parameterized by a security parameter~$\lambda$, where~$q(\lambda)$ and~$n(\lambda)$ are polynomially bounded, and~$n(\lambda) = \omega(\log^2\lambda)$.
The distribution should be samplable by a QPT algorithm, which we denote by $(p, \theta) \gets \mathcal{Z}(1^\lambda)$ and we require one of the following computational indistinguishability assumptions
(\cref{def:computational indist}):
\begin{equation}\label{eq:Z indist strong}
\left((p, \theta) : (p, \theta) \gets \mathcal{Z}(1^\lambda)\right) \approx_{sc}
\left((p, \theta^*) : (p, \cdot) \gets \mathcal{Z}(1^\lambda); \theta^* \gets \{0,1\}^n\right)
\end{equation}
or
\begin{equation}\label{eq:Z indist standard}
\left((p, \theta) : (p, \theta) \gets \mathcal{Z}(1^\lambda)\right) \approx_c
\left((p, \theta^*) : (p, \cdot) \gets \mathcal{Z}(1^\lambda); \theta^* \gets \{0,1\}^n\right).
\end{equation}
In the game that we are about to define, $p$ models public parameters revealed to the players before they have to agree on a joint quantum state, while the value~$\theta$ is only revealed afterwards.

\begin{definition}[Computational Monogamy-of-Entanglement Game]\label{def:game}
Given a distribution~$\mathcal Z$ as above, we define the following \emph{computational monogamy-of-entanglement game} between Alice and a pair of ``adversaries'' Bob and Charlie.
It is parametrized by a security parameter~$\lambda$ and consists~of~four~phases:
\begin{enumerate}
\item \emph{Sampling Phase:} Alice samples $(p, \theta) \gets \mathcal{Z}(1^\lambda)$ and reveals $p$ to Bob and Charlie.
\item \emph{Efficient Preparation Phase:} Alice, Bob, and Charlie jointly apply a QPT algorithm (which may depend~$p$ but not on~$\theta$) to create a shared quantum state between their registers $A$, $B$,~$C$.
Registers~$A$ and $B$ should consist of $n(\lambda)$ qubits, while $C$ can be arbitrary.
\item \emph{Question Phase:} Alice measures register~$A$ in the $\theta$-basis to obtain an outcome~$K_A$.
She then reveals $\theta$ to Bob and Charlie.
\item \emph{Semi-Honest Answer Phase:} Bob measures register~$B$ in the $\theta$-basis to obtain an outcome~$K_B$, while Charlie can apply an arbitrary (possibly inefficient) measurement of register~$C$ to obtain an outcome~$K_C$.
\end{enumerate}
The players win the game if~$K_A = K_B = K_C$.
\end{definition}

Thus a \emph{strategy} for the above game consists of a QPT algorithm that on input~$p$ outputs a state~$\rho_{ABC}^{(p)}$ (the result of the preparation phase), along with a family of (possibly inefficient) POVMs~$\{Q_C^{(k_E|p,\theta)}\}_{k_E}$ that correspond to Charlie's measurement for a given value of~$p$ and~$\theta$.
Without loss of generality we may assume that this POVM does not explicitly depend on~$p$, i.e.,~$Q_C^{(k_E|p,\theta)} = Q_C^{(k_E|\theta)}$ (indeed, $p$ can always be stored in~$C$ during the preparation phase).
Then the winning probability of the game is given by
\begin{align}\label{eq:pwin}
  \!\!\!p_\text{win} = \Pr(K_A\!=\!K_B\!=\!K_C) = \!\!\underset{(p, \theta) \gets \mathcal{Z}(1^\lambda)}{\mathsf{E}}\! \sum_k \Tr\left( ({}_\theta\!\ketbra{kk}_\theta \ot Q^{(k|\theta)}_C) \rho_{ABC}^{(p)} \right).
\end{align}

%-----------------------------------------------------------------------------
\subsection{Bound on the Min-Entropy and the Winning Probability}
%-----------------------------------------------------------------------------
We now analyze the winning probability of the above game.
We first prove a slightly stronger statement -- an explicit bound on the min-entropy of $K_A = K_B$ if the two keys agree (which is a necessary condition in order to win the game) given Charlie's quantum system -- and then deduce a bound on the winning probability as a corollary.

\begin{theorem}\label{thm:moe strong}
Let $\mathcal Z$ be any distribution satisfying \cref{eq:Z indist strong} with a negligible function~$\eta(\lambda)$.
For any QPT algorithm modeling the preparation phase, let us run the computational monogamy-of-entanglement game until right before Charlie's measurement.
If $K_A \neq K_B$, sample $K\gets\{0,1\}^{n(\lambda)}$ independently and uniformly at random, else set $K := K_A = K_B$.
Let~$\rho_{KC\Theta}$ denote the resulting cq-state describing the random variables~$K$ and~$\theta$ and Charlie's register~$C$.
Then, there exists $\lambda_0$ such that, for all~$\lambda\geq\lambda_0$,
\[
H_{\min}(K|C)_\rho \geq H_{\min}(K|C\Theta)_\rho \geq t(\lambda) := -\log ( \tilde O\mleft(2^{-\frac12\sqrt{n(\lambda)}}\mright) + \eta(\lambda) ). \]
In particular, $H_{\min}(K|C)_\rho = \omega(\log\lambda)$.
\end{theorem}
\begin{proof}
For notational simplicity we assume that $n(\lambda)$ is a square.
The first inequality is known as the data-processing inequality for the min-entropy and is easy to see in the cq case.
Thus we need only to prove the second inequality.
In view of \cref{thm:cq-min-entropy}, this means that we wish to prove that there exists $\lambda_0$ such that, for all $\lambda \geq \lambda_0$ and for every POVM $\{E_{C\Theta}^{(k)}\}$, we have
\begin{equation}\label{eq:goal}
  \Tr\left(\rho_{KC\Theta} \sum_k \ketbra{k}{k}_K \otimes E_{C\Theta}^{(k)}\right)
\leq \frac{1}{2^{t(\lambda)}}.
\end{equation}
Because the state~$\rho$ is classical on register~$\Theta$, we may assume that $E_{C\Theta}^{(k)} = \sum_\theta E_C^{(k|\theta)} \ot \ketbra\theta_\Theta$, where~$\{E_C^{(k|\theta)}\}_k$ is a POVM for every fixed value of~$\theta$.

Let $\rho_{ABC}^{(p)}$ denote the joint quantum state of Alice, Bob, and Charlie right before the question phase of the game, for a fixed value of~$p$, and let~$\rho_{KC}^{(p,\theta)}$ denote the cq-state defined as in the statement of the theorem, but for fixed values of~$p$ and~$\theta$.
Then,
$\rho_{KC\Theta} = \mathsf{E}_{(p, \theta) \gets \mathcal Z(1^\lambda)}
( \rho_{KC}^{(p, \theta)} \ot \ketbra\theta_\Theta )$,
so that
\begin{align}\label{eq:KC as exp}
    \Tr\left(\!\rho_{KC\Theta} \sum_k \ketbra{k}{k}_K \otimes E_{C\Theta}^{(k)}\!\right)
= \!\!\!\underset{(p, \theta) \gets \mathcal Z(1^\lambda)}{\mathsf{E}}
\Tr\left(\!\rho_{KC}^{(p,\theta)} \sum_k \ketbra{k}{k}_K \otimes E_C^{(k|\theta)}\!\right)
\end{align}
Moreover, we have
\begin{equation}\label{eq:two terms}
\begin{aligned}
  \rho_{KC}^{(p,\theta)}
&= \sum_k \ketbra{k}_K \ot \Tr_{AB}\left( ({}_\theta\!\ketbra{kk}_\theta \ot I_C) \rho_{ABC}^{(p)} \right) \\
&+ \tau_K \ot \Tr_{AB}\left(( \sum_{k_A \neq k_B} {}_\theta\!\ketbra{k_Ak_B}_\theta \ot I_C) \rho_{ABC}^{(p)} \right),
\end{aligned}
\end{equation}
where $\tau_K$ denotes the maximally mixed state on~$K$ and~$\ket{k_Ak_B}_\theta := \ket{k_A}_\theta\ket{k_B}_\theta$.
Choose any function~$s(\lambda)$ such that $s(\lambda) = \omega(\log\lambda)$ and ${n(\lambda)}/{s(\lambda)} = \omega(\log\lambda)$, with both~$s(\lambda)$ and~${n(\lambda)}/{s(\lambda)}$ integers.
Then we can define the projections
\begin{align}\label{eq:M proj}
    M_{AB}^{(1)} = \left(I -\ketbra{\phi^+}{\phi^+}^{\otimes s(\lambda)}\right)^{\otimes \frac{n(\lambda)}{s(\lambda)}}
\qquad\text{and}\qquad
    M_{AB}^{(0)} = I - M_{AB}^{(1)}
\end{align}
(where $\ket{\phi^+}$ is a single EPR pair shared between Alice and Bob), using which we further decompose the right-hand side of \cref{eq:two terms} into three terms:
\begin{equation}\label{eq:three terms}
\begin{aligned}
  \rho_{KC}^{(p,\theta)}
&= \sum_k \ketbra{k}_K \ot \Tr_{AB}\left( ({}_\theta\!\ketbra{kk}_\theta M_{AB}^{(0)} \ot I_C) \rho_{ABC}^{(p)} \right) \\
&+ \sum_k \ketbra{k}_K \ot \Tr_{AB}\left( ({}_\theta\!\ketbra{kk}_\theta M_{AB}^{(1)} \ot I_C) \rho_{ABC}^{(p)} \right) \\
&+ \tau_K \ot \Tr_{AB}\left(( \sum_{k_A \neq k_B} {}_\theta\!\ketbra{k_Ak_B}_\theta \ot I_C) \rho_{ABC}^{(p)} \right),
\end{aligned}
\end{equation}
Thus,
\begin{align*}
    \Tr\left(\rho_{KC}^{(p,\theta)} \sum_k \ketbra{k}{k}_K \otimes E_C^{(k|\theta)}\right)
&= \sum_k \Tr\left( ({}_\theta\!\ketbra{kk}_\theta M_{AB}^{(0)} \ot E_C^{(k|\theta)}) \rho_{ABC}^{(p)} \right) \\
&+ \sum_k \Tr\left( ({}_\theta\!\ketbra{kk}_\theta M_{AB}^{(1)} \ot E_C^{(k|\theta)}) \rho_{ABC}^{(p)} \right) \\
&+ \frac1{2^{n(\lambda)}} \Tr\left(( \sum_{k_A \neq k_B} {}_\theta\!\ketbra{k_Ak_B}_\theta \ot I_C) \rho_{ABC}^{(p)} \right) \\
&\!\!\!\!\!\!\!\!\!\!\!\!\!\!\!\!\!\!\!\!\leq \sqrt{\frac{n(\lambda) / s(\lambda)}{2^{s(\lambda)}}}
+ \sum_k  \Tr\left( {}_\theta\!\ketbra{kk}_\theta M_{AB}^{(1)} \rho_{AB}^{(p)} \right)
+ \frac1{2^{n(\lambda)}},
\end{align*}
where we bound the first term using \cref{lem:fixedThetaTripleX} below, for the middle term we use~$E_C^{(k|\theta)} \leq I_C$, and for the last term~$\sum_{k_A \neq k_B} {}_\theta\!\ketbra{k_Ak_B}_\theta \leq I_{AB}$.
Taking the expectation as in \cref{eq:KC as exp}, we find that
\begin{align}
\nonumber
    \Tr\left(\rho_{KC\Theta} \sum_k \ketbra{k}{k}_K \otimes E_{C\Theta}^{(k)}\right)
&= \mathsf{E}_{(p, \theta) \gets \mathcal Z(1^\lambda)} \left( \sum_k  \Tr\left( {}_\theta\!\ketbra{kk}_\theta M_{AB}^{(1)} \rho_{AB}^{(p)} \right) \right) \\
\label{eq:almost done}
&+ \sqrt{\frac{n(\lambda) / s(\lambda)}{2^{s(\lambda)}}} + \frac1{2^{n(\lambda)}},
\end{align}
In \cref{randomTheta} below we show that
\begin{align}\label{eq:ideal}
    \mathsf{E}_{(p, \cdot) \gets \mathcal Z(1^\lambda); \theta^* \leftarrow \{0,1\}^{n(\lambda)}} \left( \sum_k  \Tr\left( {}_{\theta^*}\!\ketbra{kk}_{\theta^*} M_{AB}^{(1)} \rho_{AB}^{(p)} \right) \right)
    \leq \frac1{2^{n(\lambda) / s(\lambda)}}.
\end{align}
We claim that the computational indistinguishability in \cref{eq:Z indist strong} implies that there exists $\lambda_0$, depending only on the QPT algorithm modeling the preparation phase, such that, for all~$\lambda \geq \lambda_0$,
\begin{equation}\label{eq:reduction}
\left\lvert
\begin{aligned}
&\mathsf{E}_{(p, \theta) \gets \mathcal Z(1^\lambda)} \left( \sum_k  \Tr\left( {}_\theta\!\ketbra{kk}_\theta M_{AB}^{(1)} \rho_{AB}^{(p)} \right) \right) \\
&- \mathsf{E}_{(p, \cdot) \gets \mathcal Z(1^\lambda); \theta^* \leftarrow \{0,1\}^{n(\lambda)}} \left( \sum_k  \Tr\left( {}_{\theta^*}\!\ketbra{kk}_{\theta^*} M_{AB}^{(1)} \rho_{AB}^{(p)} \right) \right)
\end{aligned} \right\rvert \leq \eta(\lambda).
\end{equation}
Indeed we can define a reduction as follows:
On input~$(p,\theta)$, simulate the efficient preparation phase (phase 2) to obtain the state~$\rho_{AB}^{(p)}$ of Alice and Bob's qubits.
Next, the apply the efficient projective measurement~$\{M_{AB,0},M_{AB,1}\}$ defined in \cref{eq:M proj}.
If the outcome is ``0'', output an arbitrary.
If the outcome is ``1'', measure Alice and Bob's qubits in the~$\theta$-basis and return ``1'' if and only if the measurement outcomes agree.
Note that the reduction so defined is efficient (the possibly inefficient POVM is not used in the reduction).
Moreover, the bias of this reduction is precisely  the left-hand side of \cref{eq:reduction}.%
\footnote{Note that, for every fixed $\theta$, the projections $\sum_k {}_{\theta}\!\ketbra{kk}_{\theta}$ and $M_{AB}^{(1)}$ commute. This follows from~\cref{lem:support}.}
Thus, \cref{eq:reduction} must hold, for otherwise we would obtain a contradiction to the computational indistinguishability assumption in \cref{eq:Z indist strong}, and $\lambda_0$ only depends on the preparation phase.
Combining \cref{eq:almost done,eq:ideal,eq:reduction}, and choosing $s(\lambda) = \sqrt{n(\lambda)}$, we obtain the upper bound~\eqref{eq:goal}:
we have, for all $\lambda \geq \lambda_0$,
\begin{align*}
  \Tr\left(\rho_{KC\Theta} \sum_k \ketbra{k}{k}_K \otimes E_{C\Theta}^{(k)}\right)
&\leq \frac1{2^{n(\lambda) / s(\lambda)}} + \eta(\lambda) + \sqrt{\frac{n(\lambda) / s(\lambda)}{2^{s(\lambda)}}} + \frac1{2^{n(\lambda)}} \leq 2^{-t(\lambda)}.
\end{align*}
Because $\lambda_0$ does not depend on the choice of POVM~$\{E_{C\Theta}^{(k)}\}$, it follows that
$p_{\text{guess}}(K|C\Theta)_\rho \leq 2^{-t(\lambda)}$.
Using \cref{thm:cq-min-entropy}, we conclude that
\begin{equation*}
H_{\min}(K|C\Theta)_\rho =  -\log p_{\text{guess}}(K|C\Theta)_\rho \geq t(\lambda). \qedhere
\end{equation*}
\end{proof}

\begin{corollary}\label{cor:negl strong}
Let $\mathcal Z$ be any distribution satisfying \cref{eq:Z indist strong} with a negligible function~$\eta(\lambda)$.
Then, for any strategy for the computational monogamy-of-entangle\-ment game, there is~$\lambda_0$ such that, for all $\lambda\geq\lambda_0$, the winning probability~is~bounded by
$\tilde O(2^{-\frac12\sqrt{n(\lambda)}}) +\eta(\lambda)$.
In particular, the winning probability~is~negligible.
\end{corollary}
\begin{proof}
Let $\rho_{ABC}^{(p)}$ denote the joint quantum state of Alice, Bob, and Charlie right before the question phase of the game, for a fixed value of~$p$.
For a fixed value of~$p$ and~$\theta$, the joint state of the random variables~$K_A, K_B$ and Charlie's register~$C$ right before Charlie's measurement is given by
\begin{align*}
  \rho_{K_A K_B C}^{(p,\theta)} = \sum_{k_A, k_B} \ketbra{k_A k_B}_{K_AK_B} \ot \sigma_C^{(p,\theta,k_Ak_B)},
\end{align*}
where
\begin{align*}
   \sigma_C^{(p,\theta,k_Ak_B)} = \Tr_{AB}\left( ({}_\theta\!\ketbra{k_Ak_B}_\theta \ot I_C) \rho_{ABC}^{(p)} \right),
\end{align*}
with $\ket{k_Ak_B}_\theta := \ket{k_A}_\theta\ket{k_B}_\theta$.
Let $\{Q_C^{(k_E|\theta)}\}_{k_E}$ denote the POVM applied by Charlie in the answer phase for a given value of~$\theta$ (as discussed below \cref{def:game} we may assume without loss of generality that this POVM does not depend explicitly on~$p$).
Then the winning probability is given by \cref{eq:pwin}:
\begin{align*}
  p_\text{win}
&= \mathsf{E}_{(p, \theta) \gets \mathcal{Z}(1^\lambda)} \sum_k \Tr\left( ({}_\theta\!\ketbra{kk}_\theta \ot E^{(k|\theta)}_C) \rho_{ABC}^{(p)} \right) \\
&= \mathsf{E}_{(p, \theta) \gets \mathcal{Z}(1^\lambda)} \sum_k \Tr\left( \sigma_C^{(p,\theta,kk)} E^{(k|\theta)}_C \right).
\end{align*}
On the other hand, the state~$\rho_{KC\Theta}$ in the statement of \cref{thm:moe strong} is given by
\begin{align*}
  \rho_{KC\Theta} = \mathsf{E}_{(p, \theta) \gets \mathcal{Z}(1^\lambda)} \left( \sum_k \ketbra k_K \ot \sigma_C^{(p,\theta,kk)} \ot \ketbra\theta_\Theta + \tau_K \ot \sum_{k_A \neq k_B} \sigma_C^{(p,\theta,k_Ak_B)} \ot \ketbra\theta_\Theta \right),
\end{align*}
where $\tau_K$ is the maximally mixed state.
Defining the POVM~$E_{C\Theta}^{(k)} := \sum_\theta E_C^{(k|\theta)} \ot \ketbra\theta_\Theta$, we see that
\begin{align*}
    p_\text{win}
\leq \tr\left( \rho_{KC\Theta} ( \sum_k \ketbra k_K \ot E_{C\Theta}^{(k)} ) \right)
\leq p_\text{guess}(K|C\Theta)_\rho
= 2^{-H_{\min}(K|C\Theta)_\rho},
\end{align*}
where the last step is due to \cref{thm:cq-min-entropy}.
Thus \cref{thm:moe strong} implies the claim.
\end{proof}

For the standard notion of computational indistinguishability, an easy adaption of these proofs yield the following variants of \cref{thm:moe strong,cor:negl strong}.

\begin{theorem}\label{thm:moe standard}
Let $\mathcal Z$ be any parameterized distribution satisfying \cref{eq:Z indist standard}.
For any QPT algorithm modeling the preparation phase, let us run the computational monogamy-of-entanglement game until right before Charlie's measurement.
If $K_A \neq K_B$, sample $K\gets\{0,1\}^{n(\lambda)}$ independently and uniformly at random, else set $K := K_A = K_B$.
Let~$\rho_{KC\Theta}$ denote the resulting cq-state describing the random variables~$K$ and~$\theta$ and Charlie's register~$C$.
Then there exists a function~$t(\lambda) = \omega(\log\lambda)$ such that the following holds for all~$\lambda$:
\[
H_{\min}(K|C)_\rho \geq H_{\min}(K|C\Theta)_\rho \geq t(\lambda).
\]
\end{theorem}

The proof proceeds as the one of \cref{thm:moe strong} -- the only difference is that the negligible function~$\eta(\lambda)$ in \cref{eq:reduction} may now depend on the preparation phase, rather than just on the computational indistinguishability assumption.

\begin{corollary}\label{cor:negl standard}
Let $\mathcal Z$ be any parameterized distribution satisfying \cref{eq:Z indist standard}.
Then, for any strategy for the computational monogamy-of-entanglement game, the winning probability is a negligible function of~$\lambda$.
\end{corollary}

%-----------------------------------------------------------------------------
\subsection{Technical Lemmas}
%-----------------------------------------------------------------------------
We now state and prove the technical lemmas used in the proof of \cref{thm:moe strong,thm:moe standard}.

\begin{lemma}\label{lem:fixedThetaTripleX}
 Let $\rho_{ABE}$ be a quantum state, where~$A$ and $B$ are $n$-qubit registers, let~$\{ Q^{(x)}_E \}_{x \in \{0,1\}^n}$ be a POVM, and let~$s$ be a divisor of~$n$.
 Then the following holds for any fixed~$\theta \in \{0,1\}^n$:
    \[ \Tr\left( \left( \sum_{x\in \{0,1\}^n}{}_\theta\!\ketbra{xx}_{\theta} \otimes Q^{(x)}_E \right) \left( I - \left(I-\ketbra{\phi^{+}}{\phi^{+}}^{\otimes s}\right)^{\otimes (n/s)}_{AB} \right) \rho_{ABE} \right) \leq \sqrt{\frac{n/s}{2^s}},
  \]
  where $\ket{xx}_\theta := \ket x_\theta \ket x_\theta$.
\end{lemma}
\begin{proof}
By Naimark's theorem, any POVM~$\{Q^{(x)}_E\}$ can be implemented by an isometry~$V_{E \to E'F}$, where $E'$ is an $n$-qubit system and~$F$ another quantum system, followed by a measurement of~$E'$ in the standard basis.
Thus we may assume without loss of generality that~$E$ is an $n$-qubit register and $Q^{(x)}_E = \ketbra x_E$.
To prove the claim, it suffices to bound the operator norm of
    \begin{align*}
        \sum_{x\in \{0,1\}^n}{}_\theta\!\ketbra{xx}_{\theta} \left( I - \left(I-\ketbra{\phi^{+}}{\phi^{+}}^{\otimes s}\right)^{\otimes (n/s)}_{AB} \right) \otimes \ketbra x_E
    \end{align*}
    Because this is an operator controlled on $E$, this norm is the maximum operator norm of
    \begin{align*}
        M_x := {}_\theta\!\ketbra{xx}_{\theta} \left( I - \left(I-\ketbra{\phi^{+}}{\phi^{+}}^{\otimes s}\right)^{\otimes (n/s)}_{AB} \right)
    \end{align*}
    for $x\in\{0,1\}^n$.
    By \cref{lmm:IminusPitensor}, we have
    \begin{equation*}
        \begin{aligned}
            P &:= \left( I - \left(I-\ketbra{\phi^{+}}{\phi^{+}}^{\otimes s}\right)^{\otimes (n/s)}_{AB} \right)\\
        &\leq
        \ketbra{\phi^{+}}{\phi^{+}}^{\otimes s} \otimes I \otimes \dots \otimes I +
        \dots +
        I \otimes \dots \otimes I \otimes \ketbra{\phi^{+}}{\phi^{+}}^{\otimes s}.
        \end{aligned}
    \end{equation*}
    Using that $P$ is a projection, we have
    \begin{align*}
        \norm{M_x}^2
    = \norm{M_x M_x^\dagger}
    = {}_\theta\!\bra{xx} P \ket{xx}_{\theta}
    \leq \sum_{j=1}^{n/s} {}_{\theta_j}\!\bra{x_jx_j} \ketbra{\phi^{+}}{\phi^{+}}^{\otimes s} \ket{x_jx_j}_{\theta_j}
    = \frac{n/s}{2^s},
    \end{align*}
    where $x_j, \theta_j \in \{0,1\}^s$ denotes the $j$-th substring of~$x$ and~$\theta$, respectively, of length~$s$.
\end{proof}

\begin{lemma}\label{lmm:proj}
We have:
\begin{align}
    \mathsf{E}_{\theta \leftarrow \{0,1\}^n} \sum_{x\in \{0,1\}^n} {}_\theta\!\ketbra{x x}_\theta
&\equiv
\left( \mathsf{E}_{\theta \leftarrow \{0,1\}} \sum_{x\in \{0,1\}} {}_\theta\!\ketbra{x x}_\theta \right)^{\otimes n} \label{eq:average-theta}\\
&=
\left(\ketbra{\phi^{+}}{\phi^{+}} + \frac{1}{2}\ketbra{\phi^{-}}{\phi^{-}} + \frac{1}{2}\ketbra{\psi^{+}}{\psi^{+}}\right)^{\otimes n}, \nonumber
\end{align}
where $\ket{xx}_\theta := \ket x_\theta \ket x_\theta$ and we use $\equiv$ to indicate that the equality holds up to the natural reordering of the systems.
\end{lemma}
\begin{proof}
We prove the lemma for $n=1$, and the general result follows since
both the left-hand side and the right-hand side of \cref{eq:average-theta} are the $n$-th tensor power of it.
To this end we use:
\begin{align*}
    \ketbra{0}{0} = \frac{I+Z}{2}, && \ketbra{1}{1} = \frac{I-Z}{2}, && \ketbra{+}{+} = \frac{I+X}{2}, && \ketbra{-}{-} = \frac{I-X}{2}.
\end{align*}
Thus:
\begin{align*}
    \mathsf{E}_{\theta \leftarrow \{0,1\}} \sum_{x\in \{0,1\}} {}_\theta\!\ketbra{xx}_\theta
&=  \frac12 \left( \ketbra{00} + \ketbra{11} + \ketbra{++} + \ketbra{--} \right) \\
&=  \frac18 \left( (I + Z)^{\otimes 2} + (I - Z)^{\otimes 2} + (I + X)^{\otimes 2} + (I - X)^{\otimes 2} \right) \\
&=  \frac14 \left( I \otimes I + Z \otimes Z + I \otimes I + X \otimes X \right) \\
&=  \ketbra{\phi^+} + \frac12\ketbra{\phi^-} + \frac12 \ketbra{\psi^+},
\end{align*}
In the last step we used that
\begin{align*}
  \frac12 \left( I \otimes I + Z \otimes Z \right)
= \ketbra{\phi^+} + \ketbra{\phi^-}, \\
  \frac12 \left( I \otimes I + X \otimes X \right)
= \ketbra{\phi^+} + \ketbra{\psi^+},
\end{align*}
which holds because $Z \otimes Z$ acts by $+1$ on the Bell states $\ket{\phi^\pm}$ and by $-1$ on the Bell states~$\ket{\psi^\pm}$, and similarly for $X \otimes X$.
This concludes the proof.
\end{proof}

\begin{lemma}\label{randomTheta}
 Let $\rho_{AB}$ be a quantum state on $n$-qubit registers $A$ and $B$, and let~$s$ be a divisor of~$n$.
 Then:
  \[ \mathsf{E}_{\theta\leftarrow \{0,1\}^n} \Tr\left(\sum_{x\in \{0,1\}^n}{}_\theta\!\ketbra{xx}{xx}_{\theta} \left(I-\ketbra{\phi^{+}}{\phi^{+}}^{\otimes s}\right)^{\otimes (n/s)}\rho_{AB}\right) \leq \frac{1}{2^{n/s}},
  \]
  where $\ket{xx}_\theta := \ket x_\theta \ket x_\theta$.
\end{lemma}
\begin{proof}
By \cref{lmm:proj}, we can rewrite
\[
\mathsf{E}_{\theta \in \{0,1\}^n} \sum_{x\in \{0,1\}^n}{}_\theta\!\ketbra{xx}{xx}_{\theta} =
\left(\ketbra{\phi^{+}}{\phi^{+}} + \frac{1}{2}\ketbra{\phi^{-}}{\phi^{-}} + \frac{1}{2}\ketbra{\psi^{+}}{\psi^{+}}\right)^{\otimes n},
\]
and therefore
 \begin{align*}
&\mathsf{E}_{\theta\leftarrow \{0,1\}^n} \sum_{x\in \{0,1\}^n}{}_\theta\!\ketbra{xx}{xx}_{\theta} \left(I-\ketbra{\phi^{+}}{\phi^{+}}^{\otimes s}\right)^{\otimes (n/s)}\\
&\quad= \left(\ketbra{\phi^{+}}{\phi^{+}} + \frac{1}{2}\ketbra{\phi^{-}}{\phi^{-}} + \frac{1}{2}\ketbra{\psi^{+}}{\psi^{+}}\right)^{\otimes n}\left(I-\ketbra{\phi^{+}}{\phi^{+}}^{\otimes s}\right)^{\otimes (n/s)}\\
&\quad= \left(\left(\ketbra{\phi^{+}}{\phi^{+}} + \frac{1}{2}\ketbra{\phi^{-}}{\phi^{-}} + \frac{1}{2}\ketbra{\psi^{+}}{\psi^{+}}\right)^{\otimes s} -\ketbra{\phi^{+}}{\phi^{+}}^{\otimes s}\right)^{\otimes (n/s)}\\
&\quad\leq \frac{1}{2^{n/s}} I.
\end{align*}
To see the latter inequality, note that $P_0 := \ketbra{\phi^+}$ and $P_1 = \ketbra{\phi^-} + \ketbra{\psi^+}$ are projectors with orthogonal range, and hence
\begin{align*}
  (P_0 + \frac{1}{2} P_1)^{\otimes s} - P_0^{\otimes s}
= \sum_{0 \neq x \in \{0,1\}^s} \frac1{2^{|x|}} P_{x_1} \otimes \cdots \otimes P_{x_s}
\leq \frac12 \sum_{0 \neq x \in \{0,1\}^s} P_{x_1} \otimes \cdots \otimes P_{x_s}
\leq \frac12 I.
\end{align*}
Therefore,
\begin{align*}    \Tr\left(\mathsf{E}_{\theta\leftarrow \{0,1\}^n} \sum_{x\in \{0,1\}^n}{}_\theta\!\ketbra{xx}{xx}_{\theta} \left(I-\ketbra{\phi^{+}}{\phi^{+}}^{\otimes s}\right)^{\otimes (n/s)}\rho_{AB}\right)
\leq \frac{1}{2^{n/s}} \Tr(\rho_{AB})
= \frac{1}{2^{n/s}},
\end{align*}
concluding our proof.
\end{proof}

%% file: protocol.tex
% !TEX root = everlasting-qkd.tex
%=============================================================================
\section{Quantum Key Distribution Protocols}\label{sec:protocol}
%=============================================================================
We first construct our non-interactive QKD protocol and establish a weak form of everlasting security (\cref{sec:niqkd}).
We then show how to convert this protocol into a two-round protocol to achieve the standard definition of everlastingly security (\cref{sub:tworound}).

%-----------------------------------------------------------------------------
\subsection{Non-Interactive Protocol with Weak Everlasting Security}\label{sec:niqkd}
%-----------------------------------------------------------------------------
We now present a quantum key distribution protocol that is \emph{non-interactive}, i.e., it consists of a single round of simultaneous messages between Alice and Bob.
Our construction assumes the existence of a post-quantum non-interactive key exchange (NIKE) protocol and upgrades it to a non-interactive quantum key distribution protocol that satisfies a weak version of everlasting security, which we will define below.

\begin{definition}[Non-Interactive QKD Protocol]\label{def:niqkd}
Let $\mathsf{NIKE}= (\mathsf{Stp},\mathsf{Gen},\mathsf{SdK})$ be a post-quantum secure NIKE protocol (\cref{def:nike}), with key space $\{0,1\}^{n(\lambda)}$, where $n(\lambda)$ grows polynomially in the security parameter~$\lambda$.
We define the following \emph{non-interactive~QKD~protocol} with the same key space $\mathsf{SKS} =  \{0,1\}^{n(\lambda)}$:
\begin{enumerate}
\item\label{step:setup} \emph{Setup:}
Run $\mathsf{pp} \leftarrow \mathsf{Stp}(1^\lambda)$.
We assume that $\mathsf{pp}$ are given as input to all parties.
\item \emph{Alice:}
Run $(\mathsf{sk}_A, \mathsf{pk}_A) \leftarrow \mathsf{Gen}(\mathsf{pp}, A)$ and prepare $n(\lambda)$ EPR pairs.
Send the classical bitstring~$\mathsf{pk}_A$ to Bob, along with one qubit of each EPR pair.

\emph{Bob:}
Sample $(\mathsf{sk}_B, \mathsf{pk}_B) \leftarrow \mathsf{Gen}(\mathsf{pp}, B)$ and send the classical bitstring $\mathsf{pk}_B$ to Alice.
\item \label{step:measurement} \emph{Output:}
Alice computes $\theta_A \leftarrow \mathsf{SdK}(B, \mathsf{pk}_B, A, \mathsf{sk}_A)$ and measures her remaining $n(\lambda)$ qubits in the $\theta_A$-basis to obtain $K_A \in \{0,1\}^{n(\lambda)}$.
Similarly, Bob computes $\theta_B \leftarrow \mathsf{SdK}(A, \mathsf{pk}_A, B, \mathsf{sk}_B)$ and measures his $n(\lambda)$ qubits in the $\theta_B$-basis to obtain $K_B \in \{0,1\}^{n(\lambda)}$.
\end{enumerate}
\end{definition}
This defines a QKD protocol that is non-interactive: There is a single round of communication, consisting of one message from Alice to Bob and one from Bob to Alice, with the two messages not depending on each other.
Moreover, the correctness of protocol is immediate:
by the correctness of the NIKE protocol, it holds that $\mathsf{SdK}(B, \mathsf{pk}_B, A, \mathsf{sk}_A) = \mathsf{SdK}(A, \mathsf{pk}_A, B, \mathsf{sk}_B)$ with overwhelming probability -- and in this case, Alice and Bob measure their EPR pairs in the same basis ($\theta_A = \theta_B$), hence they obtain the same outcome $K_A = K_B$.

However, it is easy to see that the protocol does \emph{not} satisfy the standard notion of everlasting security.
Indeed, the QPT adversary can keep the~$n(\lambda)$ qubits that Alice sends to Bob and instead output one qubit each of~$n(\lambda)$ fresh EPR pairs, and also store the public keys~$\mathsf{pk}_A, \mathsf{pk}_B$.
Since the post-quantum NIKE is only computationally secure, this information suffices to (inefficiently) learn~$\theta_A = \theta_B$.
Then $K_A$ can be obtained by suitable basis measurements on the qubits that were sent by Alice and kept by the adversary, and $K_B$ can be obtained on the remaining qubits kept by the adversary.
Interestingly, for this attack it holds that $K_A = K_B$ only with negligible probability (since~$K_A$ and~$K_B$ are now independent and uniformly random).
This is no accident.
Indeed we will now show that the protocol still satisfies a form of everlasting security provided~$K_A = K_B$.
We now give a formal definition tailored to the protocol defined in \cref{def:niqkd}:

\begin{definition}[Weak Everlasting Security]\label{def:weak-search}
Consider the following experiment involving Alice, Bob, and an adversary described by an non-uniform QPT algorithm:
\begin{enumerate}[label={{\Roman*}.},ref={{\Roman*}}]
\item\label{step:weak 1} We run step 1 of \cref{def:niqkd} and also give $\mathsf{pp}$ as an input to the adversary.
\item\label{step:weak 2} We run step 2 of \cref{def:niqkd}, but instead of delivering the two messages, we first send them to the adversary, who returns a register~$B$ (and keeps an internal register~$E$).
Modify Alice's message to consist of the quantum register~$B$, along with the original classical bitstring, and deliver it to Bob.
Deliver Bob's message unchanged as it only consists of a classical bitstring.
\item\label{step:weak 3} We proceed by running step 3 of \cref{def:niqkd}.
Let $K_A,K_B$ denote Alice's and Bob's output, respectively.
If $K_A \neq K_B$, we set $K$ to be a uniformly random bitstring in~$\{0,1\}^{n(\lambda)}$.
Otherwise, we set $K := K_A = K_B$.
Let~$\rho_{KE}$ be the classical-quantum joint state of $K$ and the adversary's internal register~$E$.
\end{enumerate}
We say that the protocol satisfies \emph{weak everlasting security} if there exists a function $t(\lambda) = \omega(\log(\lambda))$ such that following holds:
for every QPT adversary, there exists $\lambda_0$ such that, for all~$\lambda \geq \lambda_0$,
\begin{equation}\label{eq:min entropy}
H_{\min}(K|E)_\rho \geq t(\lambda).
\end{equation}
\end{definition}

Using our computational monogamy of entanglement result, we now show that the QKD protocol indeed satisfies this weaker notion of everlasting security.

\begin{theorem}\label{thm:weak}
The non-interactive QKD protocol (\cref{def:niqkd}) is correct and satisfies weak everlasting security (\cref{def:weak-search}).
\end{theorem}
\begin{proof}
We already established correctness in the discussion above, so it remains to prove security.
We can write the cq-state in \cref{def:weak-search} as~$\rho_{KE} = \mathsf{E}_{\theta_A,\theta_B} \rho_{KE}^{(\theta_A\theta_B)}$, where~$\rho_{KE}^{(\theta_A\theta_B)}$ is the cq-state conditioned on fixed values of~$\theta_A$ and~$\theta_B$ and the average is over the marginal distribution of~$(\theta_A,\theta_B)$.
To establish the bound on the min-entropy, we wish to compare~$\rho_{KE}$ to the cq-state arising in the computational monogamy-of-entanglement theorem (\cref{thm:moe strong}).
Let $\mathcal Z(1^\lambda)$ denote the joint distribution of~$(p,\theta_A)$ sampled by the following efficient algorithm:
\begin{enumerate}[noitemsep,nolistsep]
    \item Sample $\mathsf{pp} \leftarrow \mathsf{Setup}(1^\lambda)$, $(\mathsf{sk}_A, \mathsf{pk}_A) \leftarrow \mathsf{Gen}(\mathsf{pp},A)$, and $(\mathsf{sk}_B, \mathsf{pk}_B) \leftarrow \mathsf{Gen}(\mathsf{pp},B)$.
    \item Output $p := (\mathsf{pp}, \mathsf{pk}_A, \mathsf{pk}_B)$ and $\theta_A := \mathsf{SdK}(B, \mathsf{pk}_B, A, \mathsf{sk}_A)$.
\end{enumerate}
The post-quantum security of the NIKE in \cref{eq:pq strong} implies the computational indistinguishability for the computational monogamy-of-entanglement game in \cref{eq:Z indist strong}.
Now suppose that Eve plays the role of Charlie ($C=E$) and let~$\rho_{ABC}^{(p)}$ denote the state of Alice, Bob, and Charlie after steps~I and~II of \cref{def:weak-search} (which are efficient).
This constitutes an efficient preparation phase for the computational monogamy-of-entanglement game.
Because in the game both Alice and Bob use the same measurement basis~$\theta_A$, the cq-state described in \cref{thm:moe strong} is given by $\tilde\rho_{KE} := \mathsf{E}_{\theta_A} \rho_{KE}^{(\theta_A\theta_A)}$.
Thus \cref{thm:moe strong} implies that there exists~$\lambda_0$ such that, for all $\lambda \geq \lambda_0$, we have
\[ H_{\min}(K|E)_{\tilde\rho} \geq t(\lambda) \]
or, equivalently,
\begin{align*}
  p_\text{guess}(K|E)_{\tilde\rho} \leq 2^{-t(\lambda)},
\end{align*}
where $t(\lambda) = \omega(\log\lambda)$ is a function that is independent of the adversary.
Note that $2^{-t(\lambda)}$ is negligible.
By the correctness of the NIKE protocol, $\mathsf{Td}(\rho_{KE}, \tilde{\rho}_{KE}) \leq \mathsf{P}(\theta_A \neq \theta_B)$ is also a negligible function independent of the adversary.
Hence the above also holds for $\rho$, concluding the proof.
\end{proof}

One can obtain an explicit min-entropy bound in \cref{eq:min entropy} by using the formula in \cref{thm:moe strong} along with a bound on the correctness of the post-quantum secure NIKE used in the construction.
E.g., if the security of the post-quantum NIKE holds with a negligible function $2^{-\Omega(\sqrt n(\lambda))}$ and correctness holds with a failure probability of $2^{-\Omega(\sqrt n(\lambda))}$, then we have $H_{\min}(K|E) = \Omega(\sqrt{n(\lambda)})$.

We remark that if in the post-quantum security of the NIKE (\cref{eq:pq strong} in \cref{def:nike}) we replace the strong computational indistinguishability~$\approx_{sc}$ by~$\approx_c$ (\cref{def:computational indist}), then weak everlasting security still holds with the right-hand side of \cref{eq:min entropy} given by a function $t=\omega(\log\lambda)$ that can now depend on the adversary.
This is still a meaningful notion of security.
However, we need the stronger notion to construct the two-round protocol that we describe next.

%-----------------------------------------------------------------------------
\subsection{Two-Round Protocol with Everlasting Security}\label{sub:tworound}
%-----------------------------------------------------------------------------
We now describe a two-round simultaneous-message protocol to achieve the standard definition of everlasting security.
This protocol builds on the one-round protocol constructed in \cref{sec:niqkd} which satisfies weak everlasting security.
We use a collision-resistant hash function to verify that $K_A = K_B$ in the second round of communication, and a seeded randomness extractor for privacy amplification.

\begin{definition}[Two-Round QKD Protocol]\label{def:2qkd}
Let $\mathsf{NIQKD}$ be the non-inter\-active QKD protocol of \cref{def:niqkd}, with key space $\{0,1\}^{n(\lambda)}$ and min-entropy bound~\eqref{eq:min entropy} given by $t(\lambda) = \omega(\log\lambda)$.
Let~$m(\lambda) := \Theta(t(\lambda))$, and choose a collision-resistant hash function $\mathbbm{H}_\lambda = \{ h : [2^{n(\lambda)}]\to[2^{m(\lambda)}] \}$ (\cref{def:CR}), as well as a seeded strong average-case $(\Theta(m(\lambda)),2^{-\Theta(m(\lambda))})$-extractor $\textsf{Ext} \colon \mathcal S_\lambda \times \{0,1\}^{n(\lambda)} \to \{0, 1\}^{m(\lambda)}$ (\cref{def:extractor}).
% This exists because m(\lambda) + 2 \log(1/2^{-Theta(m(\lambda))}) = Theta(m(\lambda)).
We first define a two-round sub-protocol:
\begin{enumerate}
\item \emph{Setup and Round~1:} Alice and Bob prun $\mathsf{NIQKD}$ to obtain $K_A, K_B \in \{0,1\}^{n(\lambda)}$.
\item \emph{Round 2:} Alice samples $\mathsf{seed}_A \gets \mathcal S_\lambda$, $h_A \gets \mathbbm{H}_\lambda$ and sends $(\mathsf{seed}_A,h_A,h_A(K_A))$ to Bob.
Bob samples~$h_B \gets \mathbbm{H}_\lambda$ and sends $(h_B,h_B(K_B))$ to Alice.
\item \emph{Output:}
Alice returns $K_A$ if $h_B(K_A) = h_B(K_B)$, and otherwise~$\bot$.
Bob returns $K_B$ if $h_A(K_A) = h_A(K_B)$, and otherwise~$\bot$.
\end{enumerate}
We now present our \emph{two-round QKD protocol} with key space $\mathsf{SKS} = \{0,1\}^{m(\lambda)}$:%
\footnote{This protocol applies the extractor to a concatenation of the two `subkeys' in a~fixed order. To obtain a fully symmetrical protocol, Alice and Bob can in the first round sample and exchange random bits $b_A, b_B \leftarrow \{0,1\}$, and in the second round use $b = b_A \oplus b_B$ to decide whether to apply the extractor to $K^0 \Vert K^1$ or $K^1 \Vert K^0$, respectively.}
\begin{itemize}
    \item \emph{Parallel Sub-Protocol Runs:} Alice and Bob run two parallel (independent) instances of the above sub-protocol, once as above and once the roles of the two parties swapped.
    Denote by $(K_A^0, K_A^1)$ the outputs of Alice and by $(K_B^0, K_B^1)$ the outputs of Bob for the two sub-protocol runs.
    Moreover, denote by $\mathsf{seed}_A$ and~$\mathsf{seed}_B$ the seeds sampled in the two sub-protocol runs.
    \item \emph{Output:}
    If $K_A^0 \neq \bot$ and $K_A^1 \neq \bot$, Alice outputs
    \begin{align*}
        K_A^* = \mathsf{Ext}(\mathsf{seed}_A\oplus \mathsf{seed}_B, K_A^0 \| K_A^1),
    \end{align*}
    and otherwise~$\bot$.
    Likewise, if $K_B^0 \neq \bot$ and $K_B^1 \neq \bot$, Bob outputs
    \begin{align*}
        K_B^* = \mathsf{Ext}(\mathsf{seed}_A\oplus \mathsf{seed}_B, K_B^0 \| K_B^1),
    \end{align*}
    and otherwise~$\bot$.
\end{itemize}
\end{definition}

We consider the following properties:
\begin{itemize}
\item \emph{Correctness:} There exists a negligible function $\mathsf{negl}$ such that
\[
    \Pr\left(K_A^* = K_B^* \neq \bot\right) \geq 1- \mathsf{negl}(\lambda).
\]
\item \emph{Everlasting Security:}
Consider the following experiment involving Alice, Bob, and an adversary described by an non-uniform interactive QPT machine:
\begin{enumerate}[label={{\Roman*}.},ref={{\Roman*}}]
\item We run the QPT setup algorithm with input $1^\lambda$ to obtain public parameters~$\textsf{pp}$, which are given as input to Alice, Bob, and the adversary.
\item We then run the interactive protocol but with the following modification:
Recall that each message consists of a classical bitstring and a quantum register.
Instead of directly delivering the messages, the adversary can intercept them and return modified quantum registers.%
\footnote{The adversary is allowed to intercept multiple message at the same time, even across different rounds of the protocol, and also only return a subset of the quantum registers at a time, as long as compatible with the causal order of the protocol.}
The messages are then delivered with those quantum registers and the original classical bitstrings, which are always left unchanged.
\item Let $K^*_A$ denote Alice's output, let $K^*_B$ denote Bob's output, and let $E$ denote the internal register of the adversary at the end of the protocol.
\end{enumerate}
We say that the protocol satisfies \emph{everlasting security} if there exists a negligible function $\mathsf{negl}$ such that the following holds:
for any QPT adversary in the above experiment, there exists~$\lambda_0$ such that, for all~$\lambda \geq \lambda_0$, if we sample $U \leftarrow \{0,1\}^{m(\lambda)}$ independently and uniformly and put
\begin{align*}
  U_A = \begin{cases} \bot & \text{ if } K^*_A = \bot \\ U & \text{ otherwise} \end{cases}
  \quad\text{and}\quad
  U_B = \begin{cases} \bot & \text{ if } K^*_B = \bot \\ U & \text{ otherwise} \end{cases},
\end{align*}
then
\begin{align*}
    \mathsf{Td}\bigl( \rho_{EK^*_A}, \rho_{E U_A} \bigr) \leq \mathsf{negl}(\lambda)
\quad\text{and}\quad
    \mathsf{Td}\bigl( \rho_{EK^*_B}, \rho_{E U_B} \bigr) \leq \mathsf{negl}(\lambda),
\end{align*}
where $\rho_{EK^*_AK^*_BU_AU_B}$ denotes the classical-quantum state describing the adversary's internal register and the random variables~$K^*_A, K^*_B,U_A, U_B$.
\item \emph{Verifiability:}
There exists a negligible function~$\mathsf{negl}$ such that the following holds:
for any QPT adversary in the above experiment, there exists~$\lambda_0$ such that, for all~$\lambda \geq \lambda_0$,
\[
    \Pr\mleft(K^*_A \neq K^*_B\mright) \leq \mathsf{negl}(\lambda).
\]
\end{itemize}

Note that the classical messages sent by the parties are honestly delivered in the experiment underlying the security definition.
This models the presence of (public) authenticated classical channels, which is a necessary assumption for security of key exchange protocols.
We refer to~\cite{MW24} for a discussion on this aspect and for other considerations on the above definition.

The correctness of our two-round protocol follows immediately from the correctness of the non-interactive protocol.
Next we consider verifiability.
Clearly, if~$K_A^0=K_B^0$ and~$K_A^1 = K_B^1$, then~$K_A^* = K_B^*$.
Now suppose that $K_A^0 \neq K_B^0$.
Then we must have $h_A^0(K_A^0) = h_A^0(K_B^0)$ or $h_B^0(K_A^0) = h_B^0(K_B^0)$, or both.
But this can only happen with negligible probability, for otherwise we would obtain a contradiction to the collision-resistance of the hash function (since the protocol runs in QPT and the keys are sampled honestly).
The case that $K_A^1 \neq K_B^1$ works identically.
We conclude that~$K_A^* = K_B^*$ with overwhelming probability.

We now sketch our argument for everlasting security.
Without loss of generality it suffices to consider an attack of the following form.
First, the attacker intercepts the simultaneous first-round messages sent by Alice and Bob, and outputs a quantum register that, along with the classical part of Alice's original message, gets delivered to Bob.%
\footnote{The case where a message first gets delivered to Alice can be analyzed in exactly the same fashion, thanks to the symmetry of our protocol.}
Next, the attacker receives Bob's second-round message (which is classical) and outputs a quantum register that, along with the classical part of Bob's first-round message, gets delivered to Alice.
The attacker records all remaining messages, which are classical and get delivered honestly.

We henceforth concentrate on the first subprotocol and we consider the second subprotocol as part of the adversary, which is possible because the two subprotocols are independent.
Then the weak everlasting security of our NI-QKD protocol (\cref{thm:weak}) shows that $H_{\min}(K^0 | E^0) \geq t(\lambda)$, where $K^0$ is defined as in the definition of weak everlasting security (in terms of~$K^0_A$ and~$K^0_B$), and $E^0$ denotes the internal state of the adversary after the first interception. The subsequent message from Bob
$(h_B^0, h_B^0(K_B^0))$ contains information about $K_B^0$, but using the chain rule for the min-entropy and the structure of the protocol, it follows that $H_{\min}(K^0 | E) \geq \Omega(t(\lambda))$, where $E = (E_0, h_B^0, h_B^0(K_B^0))$.
Crucially, Alice's seed, $\mathsf{seed}_A$, is chosen independently from $K^0_B$ (and~$K^0_A$).
Thus, if $K^0_A = K^0_B$, so $K^0_A = K^0$, the randomness extractor guarantees that $K^*_A E$ and $U_A E$ are negligibly close in trace distance.
On the other hand, if $K^0_A \neq K^0_B$ then~$K^*_A = \bot$ with overwhelming probability by the same argument used to prove verifiability, and hence the same holds.
Combining these two cases concludes the proof.

%% file: nogo.tex
% !TEX root = everlasting-qkd.tex
%=============================================================================
\section{Entanglement is Necessary for NI-QKD}
%=============================================================================
In this section we prove that any non-interactive QKD protocol where Alice and Bob derive their shared key only from classical randomness cannot be everlastingly secure (\cref{thm:need quantum}).
This is in particular the case when the quantum messages sent by Alice and Bob are unentangled with the respective sender's quantum memories.
In other words, our result implies that entanglement is required for non-interactive quantum key distribution, and at least one of Alice or Bob must have a quantum memory (\cref{cor:need entanglement}).
In contrast, multi-round QKD can achieve everlasting security (even unconditional security) by sending unentangled systems from Alice to Bob.

Throughout this section we consider a non-interactive QKD protocol with key space $\mathsf{SKS} = \{0,1\}^{m(\lambda)}$ that has the following form:
\begin{enumerate}
\item Alice efficiently samples $(R_A,S_A,\rho_{AM_A})$, where $R_A$ and $S_A$ are $\poly(\lambda)$-length bitstrings and~$\rho_{AM_A}$ is a $\poly(\lambda)$-qubit state.
Bob efficiently samples $(R_B,S_B,\sigma_{BM_B})$, where $R_B$ and~$S_B$ are $\poly(\lambda)$-length bitstrings and $\sigma_{BM_B}$ is a $\poly(\lambda)$-qubit state.
\item Alice sends the bitstring~$S_A$ and the quantum system~$M_A$ to Bob.
Simultaneously, Bob sends the bitstring~$S_B$ and the quantum system~$M_B$ to Alice.
\item Finally, Alice applies an efficient measurement (depending on $R_A, S_A, S_B$) to quantum systems~$A M_B$ to obtain the key~$K_A \in \{0,1\}^{m(\lambda)} \cup \{\bot\}$.
Likewise, Bob applies an efficient measurement (depending on $R_B, S_A, S_B$) to~$B M_A$ to obtain~$K_B \in \{0,1\}^{m(\lambda)} \cup \{\bot\}$.
\end{enumerate}

For simplicity we further assume that the protocol is \emph{perfectly correct}.
That is, we assume that in the absence of an adversary, $K_A = K_B \neq \bot$ with certainty.

\begin{definition}[NI-QKD with classically-derived keys]\label{def:classical keys}
We say that a NI-QKD protocol of the form above has \emph{classically-derived keys} if there is a function~$f$ such that $K_A = K_B = f(R_A, R_B, S_A, S_B)$.
\end{definition}

In other words, while quantum information can used in the protocol, the agreed-upon secret key is a function of the classical randomnes sampled by Alice and Bob only (rather than of any quantum randomness produced during the protocol execution). We emphasize that the function~$f$ need not be efficient nor does it have to be known to Alice or Bob -- it merely needs to exist.

\begin{theorem}\label{thm:need quantum}
No perfectly correct NI-QKD protocol with classically-derived keys can be everlastingly secure.
In fact, for any such protocol there exists an attack with an efficient online phase that uses only nondestructive quantum measurements (hence Alice and Bob still output $K_A = K_B \neq \bot$ with certainty) and an unbounded offline phase that outputs $K_A = K_B$ with constant probability.
\end{theorem}
\begin{proof}
Without loss of generality, we may assume both~$R_A$ and~$R_B$ consist of the same number $r(\lambda)=\poly(\lambda)$ of bits.
We may furthermore assume that~$R_A$ includes~$S_A$ and~$R_B$ includes~$S_B$.
Then the function $f$ in \cref{def:classical keys} will only depends on~$R_A$ and~$R_B$, i.e.\ $K_A = K_B = f(R_A, R_B)$, and the (without loss of generality projective) efficient measurements applied by Alice and Bob to obtain their keys only depend on their local randomness and the classical message sent by the other party.
We denote Alice's projective measurements by
$\{P^{(k_A|r_A,s_B)}_{A M_B}\}_{k_A \in \{0,1\}^\lambda \cup \{\bot\}}$
and Bob's projective measurements by
$\{Q^{(k_B|r_B,s_A)}_{B M_A}\}_{k_B \in \{0,1\}^\lambda \cup \{\bot\}}$.
Then the condition $K_A = K_B = f(R_A, R_B)$ means the following:
If $(R_A, S_A, \rho_{AM_A})$ and~$(R_B, S_B, \rho_{BM_B})$ are obtained by Alice and Bob's sampling algorithms, then with probability one we have
\begin{align}\label{eq:P is det}
    \tr P^{(k_A|R_A,S_B)}_{A M_B} (\rho_A \ot \rho_{M_B}) = \delta_{k_A,f(R_A, R_B)}.
\end{align}
as well as
\begin{align}\label{eq:Q is det}
    \tr Q^{(k_B|R_B,S_A)}_{B M_A} (\rho_B \ot \rho_{M_A}) = \delta_{k_B,f(R_A, R_B)}.
\end{align}
In particular, both measurement outcomes are deterministic (conditional on $R_A$ and $R_B$) and hence these projective measurements do \emph{not} change the measured quantum registers.
% This is because the equations show that~$\rho_A \ot \rho_{M_B}$ is supported in an eigenspace of the projection, etc.

With this in mind we consider the following attack:
\begin{itemize}
\item In the efficient online phase, Eve intercepts the classical $S_A$ and $S_B$~messages and the $M_A$ and~$M_B$ quantum registers sent by Alice and Bob.
For~$t \in [2r(\lambda)]$:
\begin{itemize}
    \item Eve runs Bob's sampling algorithm to obtain $(R^{(t)}_B,\rho^{(t)}_{BM_B})$ and applies Bob's measurement $\{Q^{(k_B|R^{(t)}_B,S_A)}_{B M_A}\}$ on the~$B$ register of~$\rho^{(t)}$ and the~$M_A$ register received from Alice.
    By \cref{eq:Q is det}, the measurement outcome is $\alpha_t := f(R_A, R^{(t)}_B)$, where $R_A$ is the random variable generated by Alice in step~1 of the protocol.
    \item Eve runs Alice's sampling algorithm to obtain $(R^{(t)}_A, \rho^{(t)}_{AM_A})$ and applies Alice's measurement $\{P^{(k_A|R^{(t)}_A,S_B)}_{A M_B}\}$ on the~$A$ register of~$\rho^{(t)}$ and the~$M_B$ register received from Bob.
    By \cref{eq:P is det}, the measurement outcome is $\beta_t := f(R^{(t)}_A, R_B)$, where $R_B$ is the random variable generated by Bob in step~1 of the protocol.
\end{itemize}
Finally Eve passes the quantum registers~$M_A$ and~$M_B$ to Bob and Alice, respectively.
The classical messages~$S_A$ and~$S_B$ are also delivered honestly.
\item In the inefficient offline phase, Eve defines sets~$\Gamma_A := \{ r_A^* : f(r_A^*, R^{(t)}_B) = \alpha_t \ \forall t \in [2r(\lambda)] \}$~and $\Gamma_B := \{ r_B^* : f(R^{(t)}_A, r_B^*) = \beta_t \ \forall t \in [2r(\lambda)] \}$.
First, Eve samples~$R_A^*$ from the distribution of~$R_A$ \emph{conditional on the event~$R_A \in \Gamma_A$}.
Then, Eve~picks any~$R_B^* \in \Gamma_B$, and outputs~$f(R_A^*, R_B^*)$.
\end{itemize}
Note that the quantum registers are passed on unchanged.
Therefore, Alice and Bob will still agree on the key $f(R_A,R_B)$ (not equal to $\bot$) with overwhelming probability.

To analyze the attack, we define the sets
\begin{align*}
    G_A^{(r_A, \{r_B^{(t)}\})} &:= \left\{ r^*_A \in \{0,1\}^{r(\lambda)} \;\middle|\; f(r^*_a, r^{(t)}_B) = f(r_A, r^{(t)}_B) \ \forall t\in[2r(\lambda)] \right\}, \\
    H_A^{(r_A)} &:= \left\{ r^*_A \in \{0,1\}^{r(\lambda)} \;\middle|\; \Pr_{R_B}\mleft( f(r^*_A, R_B) = f(r_A, R_B) \mright) \leq \frac23 \right\}.
\end{align*}
For any $r_A \in \{0,1\}^{r(\lambda)}$ and any $r^*_A \in H_A^{(r_A)}$, we have that
\begin{align*}
    \Pr_{\{R_B^{(t)}\}}\mleft( r^*_A \in G_A^{(r_A,\{R_B^{(t)}\})} \mright)
&= \Pr_{\{R_B^{(t)}\}}\mleft( f(r^*_a, R_B^{(t)}) = f(r_A, R_B^{(t)}) \ \forall t\in[2r(\lambda)] \mright) \\
&= \prod_{t=1}^{2r(\lambda)} \Pr_{\{R_B^{(t)}\}}\mleft( f(r^*_a, R_B^{(t)}) = f(r_A, R_B^{(t)}) \mright)
\leq \left( \frac23 \right)^{2 r(\lambda)} = \left( \frac49 \right)^{r(\lambda)}
\end{align*}
because the $R^{(t)}_B$ are sampled independently and from the same distribution as $R_B$.
As a consequence, we have for any $r_A \in \{0,1\}^{r(\lambda)}$, using the union bound,
\begin{align*}
    \Pr_{\{R_B^{(t)}\}}\mleft( G_A^{(r_A,\{R_B^{(t)}\})} \cap H_A^{(r_A)} \neq \emptyset \mright)
\leq \sum_{r^*_A \in H_A^{(r_A)}} \Pr_{\{R_B^{(t)}\}}\mleft( r^*_A \in G_A^{(r_A,\{R_B^{(t)}\})} \mright)
\leq \left( \frac89 \right)^{r(\lambda)}.
\end{align*}
As the $R^{(t)}_B$ are sampled independenly of $R_A$, the above bound also holds with~$R_A$ in place of~$r_A$.
Because $R_A^* \in \Gamma_A = G_A^{(R_A,\{R_B^{(t)}\})}$, we obtain
\begin{align}\label{eq:R_A bound}
     \Pr\left( R_A^* \in H_A^{(R_A)} \right)
\leq \Pr_{R_A, \{R_B^{(t)}\}}\mleft( G_A^{(R_A,\{R_B^{(t)}\})} \cap H_A^{(R_A)} \neq \emptyset \mright)
\leq \left( \frac89 \right)^{r(\lambda)}.
\end{align}
Similarly, if we define for $r_B \in \{0,1\}^{r(\lambda)}$ the set
\begin{align*}
    H_B^{(r_B)} &:= \left\{ r^*_B \in \{0,1\}^{r(\lambda)} \;\middle|\; \Pr_{R_A}\mleft( f(R_A, r^*_B) = f(R_A, r_B) \mright) \leq \frac23 \right\},
\end{align*}
we obtain with the same argument that
\begin{align}\label{eq:R_B bound}
    \Pr\left( R_B^* \in H_B^{(R_B)} \right) \leq \left( \frac89 \right)^{r(\lambda)}.
\end{align}

Now, on the one hand side, we have
\begin{align*}
    \Pr\mleft( f(R_A^*, R_B) = f(R_A, R_B) \mid R_A^* \not\in H_A^{(R_A)} \mright) \geq \frac23
\end{align*}
by definition of the set~$H_A^{(R_A)}$ and because $R_B$ is sampled independently from~$R_A$ and~$R_A^*$.
Together with \cref{eq:R_A bound}, we obtain
\begin{align}\nonumber
  \Pr\mleft( f(R_A^*, R_B) = f(R_A, R_B) \mright)
% \geq \Pr\mleft(f(R_A, R_B) = f(R_A^*, R_B) \mid R_A^* \not\in H_A^{(R_A)} \mright) \Pr(R_A^* \not\in H_A^{(R_A)})
% = \Pr\mleft(f(R_A, R_B) = f(R_A^*, R_B) \mid R_A^* \not\in H_A^{(R_A)} \mright) \left( 1 - \Pr(R_A^* \in H_A^{(R_A)}) \right)
&\geq \Pr\mleft(f(R_A^*, R_B) = f(R_A, R_B) \mid R_A^* \not\in H_A^{(R_A)} \mright) - \Pr(R_A^* \in H_A^{(R_A)}) \\
\label{eq:add first star}
&\geq \frac23 - \left( \frac89 \right)^{r(\lambda)}.
\end{align}
On the other hand, we have
\begin{align*}
&\quad \Pr\mleft(f(R_A^*, R_B^*) = f(R_A^*, R_B) \mid R_B^* \not\in H_B^{(R_B)} \mright) \\
&= \Pr\mleft(f(R_A, R_B^*) = f(R_A, R_B) \mid R_B^* \not\in H_B^{(R_B)} \mright)
\geq \frac23,
\end{align*}
The first equality holds because~$R_A^*$ has the same distribution as~$R_A$,%
\footnote{This is an instance of the following fact (with $X=R_A$, $Y=\Gamma_A$, $X^* = R_A^*$):
Given two random variables~$X$~and~$Y$, sample~$X^*$ from the distribution~$p(\cdot|Y)$, where~$p(x|y)=p(x,y)/p(y)$ denotes the conditional distribution of~$X$ given~$Y$.
Then $X^*$ has the same distribution as~$X$.
Indeed, $\sum_y p(y) p(x|y) = p(x)$.
}
and either random variable is independent of~$R_B$ and~$R_B^*$.
Hence we obtain as above, but using \cref{eq:R_B bound}, that
\begin{align*}
    \Pr\mleft( f(R_A^*, R_B^*) = f(R_A^*, R_B) \mright) \geq \frac23 - \left( \frac89 \right)^{r(\lambda)}.
\end{align*}
Together with \cref{eq:add first star}, we conclude that
\begin{align*}
    \Pr\mleft( f(R_A^*, R_B^*) = f(R_A, R_B) \mright)
\geq \frac13 - 2 \left( \frac89 \right)^{r(\lambda)}
= \frac13 - \mathsf{negl}(\lambda),
\end{align*}
which establishes the desired lower bound on the success probability of the attack.
\end{proof}

As is clear from the statement of the theorem, Eve's attack also works in the weak everlasting security model (\cref{sec:niqkd}).

\begin{lemma}
If a perfectly correct NI-QKD protocol of the form described at the beginning of the section is such that both~$\rho_{AM_A}$ and~$\rho_{BM_B}$ are unentangled with certainty, then the protocol has classically-derived keys in the sense of \cref{def:classical keys}.
\end{lemma}
\begin{proof}
We may assume that~$\rho_{AM_A}$ and~$\rho_{BM_B}$ are pure states, hence~$\rho_{AM_A} = \rho_A^{(R_A,S_A)} \ot \rho_{M_A}^{(R_A,S_A)}$and likewise~$\rho_{BM_B} = \rho_B^{(R_B,S_B)} \ot \rho_{MB}^{(R_B,S_B)}$.
Thus, $\rho_{AM_A} \ot \rho_{BM_B}$ is also a product state between registers~$AM_B$ and~$BM_A$.
Because Alice and Bob's keys are obtained by applying measurements (depending on $R_A,R_B,S_A,S_B$) on registers~$AM_B$ and~$BM_A$, respectively, we see that $K_A$ and $K_B$ are conditionally independent given~$R_A,R_B,S_A,S_B$.
On the other hand, we have~$K_A = K_B$ by perfect correctness.
Thus $K_A$ and $K_B$ must be deterministic and equal to each other given~$R_A,R_B,S_A,S_B$.
In other words, there exists a function $f$ such that $K_A = K_B = f(R_A,R_B,S_A,S_B)$.
\end{proof}

\begin{corollary}\label{cor:need entanglement}
If a perfectly correct NI-QKD protocol is everlastingly secure, at least one of~$\rho_{AM_A}$ and~$\rho_{BM_B}$ must be entangled.
In particular, Alice or Bob need to have a quantum memory.
\end{corollary}